\documentclass[runningheads, envcountsame, orivec]{llncs}

\usepackage{beyondn-mysty-min}

\newcommand{\uli}[1]{}
\newcommand{\gs}[1]{}

\begin{document}

\title{Generating Posets beyond \N}

\author{Uli Fahrenberg$^1$, Christian Johansen$^2$, Georg
  Struth$^3$, Ratan Badahur Thapa$^2$}

\authorrunning{Fahrenberg, Johansen, Struth and Thapa}

\institute{{\'E}cole Polytechnique, Palaiseau, France \and University
    of Oslo, Norway \and University of Sheffield, UK}

\maketitle

\begin{abstract}
  We introduce iposets---posets with interfaces---equipped with a
  novel gluing composition along interfaces and the standard parallel
  composition. We study their basic algebraic properties as well as
  the hierarchy of gluing-parallel posets generated from singletons by
  finitary applications of the two compositions. We show that not only
  series-parallel posets, but also interval orders, which seem more
  interesting for modelling concurrent and distributed systems, can be
  generated, but not all posets. Generating posets is also important
  for constructing free algebras for concurrent semirings and Kleene
  algebras that allow compositional reasoning about such
  systems.
\end{abstract}

%%%%%%%%%%%%%%%%%%%%%%%%%%%%%%%%%%%%%%%%%%%%%%%%%%%%%%%%%%%%%%%% 

\section{Introduction}

This work is inspired by Tony Hoare's programme of building graph
models of concurrent Kleene algebra
($\CKA$)~\cite{DBLP:journals/jlp/HoareMSW11} for real-world
applications. $\CKA$ extends the sequential compositions,
nondeterministic choices and unbounded finite iterations of imperative
programs modelled by Kleene algebra into concurrency, adding
operations of parallel composition and iteration, and a weak
interchange law for the sequential-parallel interaction. Such algebras
have a long history in concurrency theory, dating back at least to
Winkowski~\cite{DBLP:journals/ipl/Winkowski77}.  Commutative Kleene
algebra---the parallel part of $\CKA$---has been investigated by
Pilling and Conway~\cite{Conway71}. A double semiring with weak
interchange---$\CKA$ without iteration---has been introduced by
Gischer~\cite{DBLP:journals/tcs/Gischer88}; its free algebras have
been studied by Bloom and
\'Esik~\cite{DBLP:journals/tcs/BloomE96a}. $\CKA$, like Gischer's
concurrent semiring, has both interleaving and true concurrency
models, e.g. shuffle as well as pomset languages. Series-parallel
pomset languages, which are generated from singletons by finitary
applications of sequential and parallel compositions, form free
algebras in this
class~\cite{DBLP:journals/corr/LaurenceS17,DBLP:conf/esop/KappeB0Z18}
(at least when parallel iteration is ignored). The inherent
compositionality of algebra is thus balanced by the generative
properties of this model. Yet despite this and other theoretical work,
applications of $\CKA$ remain rare.

One reason is that series-parallel pomsets are not expressive enough
for many real-world applications: even simple producer-consumer
examples cannot be modelled~\cite{DBLP:journals/tcs/LodayaW00}.
\emph{Tests}, which are needed for the control structure of concurrent
programs and as assertions, are hard to capture in models of $\CKA$
(see~\cite{DBLP:journals/jlp/JipsenM16} and its discussion
in~\cite{DBLP:conf/concur/KappeBRSWZ19}). Finally, it remains unclear
how modal operators could be defined over graph models akin to pomset
languages, which is desirable for concurrent dynamic algebras and
logics beyond alternating
nondeterminism~\cite{DBLP:journals/jacm/Peleg87,DBLP:journals/tocl/FurusawaS15}.

A natural approach to generating more expressive pomset languages is
to ``cut across'' pomsets in more general ways when (de)composing
them. This can be achieved by (de)composing along interfaces, and this
idea can be traced back again to
Winkowski~\cite{DBLP:journals/ipl/Winkowski77}; see
also~\cite{DBLP:books/daglib/0030804, DBLP:conf/birthday/FioreC13,
  DBLP:journals/corr/Mimram15} for interface-based compositions of
graphs and posets, or~\cite{DBLP:conf/RelMiCS/HoareSMSVZO14,
  DBLP:conf/mpc/MollerH15, DBLP:conf/utp/MollerHMS16} for recent
interface-based graph models for $\CKA$.  As a side effect, interfaces
may yield notions of tests or modalities.  When they consist of
events, cutting across them presumes that they extend in time and thus
form intervals.  Interval
orders~\cite{Wiener14,journals/mpsy/Fishburn70} of events with
duration have been applied widely in partial order semantics of
concurrent and distributed systems~\cite{DBLP:journals/dc/Lamport86,
  DBLP:journals/jacm/Lamport86a, DBLP:conf/parle/GlabbeekV87,
  DBLP:conf/ifip2/Glabbeek90, DBLP:journals/dc/Vogler91,
  DBLP:books/sp/Vogler92, DBLP:journals/tcs/JanickiK93} and the
verification of weak memory
models~\cite{DBLP:journals/toplas/HerlihyW90}, yet generating them
remains an open problem~\cite{DBLP:journals/iandc/JanickiY17}.

Our main contribution lies in a new class and algebra of posets with
interfaces (\emph{iposets}) based on these ideas. We introduce a new
gluing composition that acts like standard serial po(m)set composition
outside of interfaces, yet glues together interface events, thus
composing events that did not end in one component with those that did
not start in the other one. Our definitions are categorical so that
isomorphism classes of posets are considered ab initio. Their
decoration with labels is then trivial, so that we may focus on posets
instead of pomsets.

Our main technical results concern the hierarchy of gluing-parallel
posets generated by finitary applications of this gluing composition
and the standard parallel composition of po(m)sets, starting from
singleton iposets.\footnote{There is only one singleton poset, but
  with interfaces, there are \emph{four} singleton iposets.}  It is
obvious that all series-parallel pomsets can be generated, but also
all interval orders are captured at the second alternation level of
the hierarchy. Beyond that, we show that the gluing-parallel hierarchy
does not collapse and that posets with certain zigzag-shaped induced
subposets are excluded.  Yet a precise characterisation of the
generated (i)posets remains open. Series-parallel posets, by
comparison, exclude precisely those posets with induced \N-shaped
subposets; interval orders exclude precisely those with induced
subposets $\twotwo$, which makes the two classes incomparable.
Iposets thus retain at least the pleasant compositionality properties
of series-parallel pomsets and the wide applicability of interval
orders in concurrency and distributed computing.

In addition, we establish a bijection between isomorphism classes of
interval orders and certain equivalence classes of interval
sequences~\cite{DBLP:conf/ifip2/Glabbeek90}, and we study the basic
algebraic properties of iposets, including weak interchange laws and
a Levi lemma. The relationship between gluing-parallel ipo(m)set
languages and $\CKA$ is left for another article.

\section{Posets and Series-Parallel Posets}
\label{se:posets}

A \emph{poset} $(P,\mathord{\le})$ is a set $P$ equipped with a
\emph{partial order} $\mathord\le$; a reflexive, transitive,
antisymmetric relation $\mathord\le$ on $P$.  A \emph{morphism} of
posets $P$ and $Q$ is an order-preserving function $f: P\to Q$, that
is, $x\le_P y$ implies $f( x)\le_Q f( y)$.  Posets and their morphisms
define the category $\Pos$.

A poset is \emph{linear} if each pair of elements is comparable with
respect to its order.  We write $\mathord<$ for the strict part of
$\le$. We write $[ n]$, for $n\ge 1$, for the \emph{discrete $n$-poset}
$(\{ 1,\dotsc, n\},\mathord{\le})$, which satisfies $i\le
j\Leftrightarrow i= j$. Additionally,
% We denote by
% $\langle n\rangle$ the (isomorphism class of the) \emph{topological
%   $n$-poset} $\{ 1,\dotsc, n\}$ with the standard ordering.
$[ 0]= \emptyset$.

The 
%monomorphisms in $\Pos$ are the order-preserving injections; the
isomorphisms in $\Pos$ are \emph{order bijections}: bijective
functions $f: P\to Q$ for which
$x\le_P y\Leftrightarrow f( x)\le_Q f( y)$. We write $P\cong Q$ if
posets $P$ and $Q$ are isomorphic. We generally consider posets up-to
isomorphism and assume, moreover, that all posets are finite.
 
Concurrency theory often considers (isomorphism classes of) posets
with points labelled by letters from some alphabet, which represent
actions of some concurrent system. These are known as \emph{partial
  words} or \emph{pomsets}.  As we are mainly interested in structural
aspects of concurrency, we ignore such labels.

Series-parallel posets form a well investigated class that can be
generated from the singleton poset by finitary applications of two
compositions.  Their labelled variants generalise rational languages
into concurrency. For arbitrary posets, these compositions are defined
as follows.

\begin{definition}
  Let $P_1=( P_1, \mathord{ \le_1})$ and $P_2=( P_2, \mathord{ \le_2})$ be
  posets. 
\begin{enumerate}[nosep]
\item Their \emph{serial composition} is the poset
  $P\pomser Q=( P\sqcup Q, \mathord\le_1\cup \mathord\le_2\cup
  P_1\times P_2)$. 
  % , where for
  % $v, v'\in P\cup Q$,
  % \begin{equation*}
  %   v\le v' \quad\text{iff}\quad v\le_P v'\text{ or } v\le_Q v',
  %   \text{ or } v\in P\text{ and } v'\in Q\,.
  % \end{equation*}
\item  Their \emph{parallel composition} is the poset
  $P_1\otimes P_2=( P_1\sqcup P_2,\mathord\le_1\cup \mathord\le_2)$.
  % , where for
  % $v, v'\in P\cup Q$,
  % \begin{equation*}
  %   v\le v' \quad\text{iff}\quad v\le_P v'\text{ or } v\le_Q v'\,.
  % \end{equation*}
\end{enumerate}
\end{definition}

Here, $\sqcup$ means disjoint union (coproduct) of sets.  We
generalise serial composition to a gluing composition in
Section~\ref{se:iposets}, after equipping posets with interfaces.
% The parallel composition is the coproduct (and not the tensor) in
% $\Pos$, but we keep the standard terminology and notation.

Serial and parallel compositions respect isomorphism, and $[ n+ m]$ is
isomorphic to $[ n]\otimes[ m]$ with isomorphism
$\varphi_{ n, m}:[ n+ m]\to[ n]\otimes[ m]$ given by
\begin{equation*}
  \varphi_{ n, m}( i)=
  \begin{cases}
    i_{[ n]} &\text{if } i\le n\,, \\
    ( i- n)_{[ m]} &\text{if } i> n\,.
  \end{cases}
\end{equation*}

% Every discrete poset is a parallel product of one-element posets:
% $[ n]= [ n- 1]\otimes[ 1]$.

By definition, a poset is \emph{series-parallel} (an \emph{sp-poset})
if it is either empty or can be obtained from the singleton poset by
applying the serial and parallel compositions a finite number of
times.  It is well known~\cite{DBLP:journals/siamcomp/ValdesTL82,
  DBLP:journals/fuin/Grabowski81} that a poset is series-parallel iff
it does not contain the induced subposet
$\N= \pomset{\cdot \ar[r] & \cdot \\ \cdot \ar[r] \ar[ur] &
  \cdot}$.\footnote{This means that there is no injection $f$ from
  $\N$ satisfying $x\le y\Leftrightarrow f(x)\le f(y)$.}  

Sp-po(m)sets form bi-monoids with respect to serial and parallel
composition, and with the empty poset as shared unit---in fact the
free algebras in this class. Compositionality of the recursive
definition of sp-po(m)sets is thus reflected by the compositionality
of their algebraic properties, which is often considered a desirable
property of concurrent systems~\cite{DBLP:books/sp/Vogler92}.  Yet
sp-posets are, in fact, too compositional for many applications: even
simple consumer-producer problems inevitably generate
$\N$'s~\cite{DBLP:journals/tcs/LodayaW00}, as shown in
Fig.~\ref{fi:prodcon} which contains the $\N$ spanned by
$c_1$, $c_2$, $p_2$, and $p_3$ as an induced subposet among others.

\begin{figure}[tbp]
  \centering
  \begin{tikzpicture}[->, >=latex', x=1.5cm]
    \foreach \x in {1,2,3,4} \node (p\x) at (\x,0) {$p_\x$};
    \foreach \x in {1,2,3,4} \node (c\x) at (\x,-1) {$c_\x$};
    \node (p5) at (5.2,0) {$\cdots$};
    \node (c5) at (5.2,-1) {$\cdots$};
    \foreach \x in {1,2,3,4} \path (p\x) edge (c\x);
    \foreach \i/\j in {1/2,2/3,3/4,4/5} \path (p\i) edge (p\j);
    \foreach \i/\j in {1/2,2/3,3/4,4/5} \path (c\i) edge (c\j);
    % \path (p1) edge[densely dashed] (c3);
  \end{tikzpicture}
  \caption{The producer-consumer example}
  \label{fi:prodcon}
\end{figure}
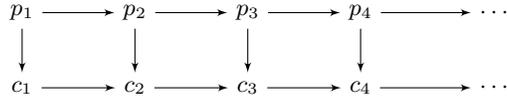

%%%%%%%%%%%%%%%%%%%%%%%%%%%%%%%%%%%%%%%%%%%%%%%%%%%%%

\section{Interval orders and  interval sequences}\label{S:interval-orders}

Interval orders~\cite{Wiener14, journals/mpsy/Fishburn70} form another
class of posets that are ubiquitous in concurrent and distributed
computing. Intuitively, they are isomorphic to sets of intervals on
the real line that are ordered whenever they do not overlap.

\begin{definition}
  An \emph{interval order} is a relational structure $(P,<)$ with $<$
  irreflexive such that $w< y$ and $x< z$ imply $w< z$ or
  $x< y$, for all $w,x,y,z\in P$.
\end{definition}
Transitivity of $<$ follows.  An alternative geometric
characterisation is that interval orders are precisely those posets
that do not contain the induced subposet
$\twotwo= \pomset{\cdot \ar[r] & \cdot \\ \cdot \ar[r] & \cdot}$.

The intuition is captured by Fishburn's
theorem~\cite{journals/mpsy/Fishburn70}, which implies that a finite
poset $P$ is an interval order iff it has an \emph{interval
  representation}: a pair of functions $b,e:P\to Q$ into some linear
order $(Q,<_Q)$ such that $b(x)<_Q e(x)$, for all $x\in P$, and
$x<_P y \Leftrightarrow e(x)<_Q b(y)$, for all $x,y\in P$. By the
first condition, pairs $(b(x),e(x))$ correspond to intervals
$I(x)=[b(x),e(x)]$ in $Q$; by the second condition, $x<_P y$ iff
$I( x)$ lies entirely before $I(y)$ in $Q$.

We write $\irep(P)$ for the set of interval representations of $P$.
Each representation can be rearranged such that all endpoints of
intervals are distinct~(\cite{GolumbicT04}, Lemma 1.5). We henceforth
assume that all interval presentations have this property.  It then
holds that $|Q|=2|P|$, and we can fix $Q$ as the target type of any
interval representation of $P$.

% Moreover, interval orders are precisely those posets in which the sets
% of downsets and upsets of elements are totally ordered with respect to
% set inclusion.  These sets can be used for constructing interval
% representations. 

Finally, with relation $\sqsubset$ on the set of maximal
antichains of  poset $P$ given by
\begin{equation*}
  A\sqsubset B \Leftrightarrow (\forall x\in A\setminus B.\forall y\in
  B\setminus A.\
  x<y),
\end{equation*}
it has been shown that $P$ is an interval order iff $\sqsubset$ is a
strict linear order~\cite{book/Fishburn85}.

Interval orders also occur implicitly in the ST-traces of Petri
nets~\cite{DBLP:conf/ifip2/Glabbeek90}. In a pure order-theoretic
setting, these are \emph{interval sequences}, that is,
% equivalence classes of
sequences of $b(x)$ and $e(x)$, with $x$ from some finite set $P$, in
which each $b(x)$ occurs exactly once and each $e(x)$ at most once and
only after the corresponding $b(x)$. An interval sequence is
\emph{closed} if each $e(x)$ occurs exactly
once~\cite{DBLP:conf/ifip2/Glabbeek90,DBLP:books/sp/Vogler92}.  An
\emph{interval trace}~\cite{DBLP:journals/iandc/JanickiY17} is an
equivalence class of interval sequences modulo the relations
$b(x)b(y) \approx b(y)b(x)$ and $e(x)e(y) \approx e(y)e(x)$ for all
$x,y\in P$. We write $\approx^\ast$ for the congruence generated by
$\approx$ on interval sequences. We identify interval sequences and
interval traces with the Hasse diagrams of their linear orders over
$Q$.

\begin{lemma}
  Let $P$ be an interval order and $(b,e)\in \irep(P)$.  Then
  $(Q, <_Q)$ is a closed interval sequence.
\end{lemma}
\begin{proof}
  Trivial.\qed
\end{proof}
We write $\sigma_{(b,e)}(P)$ for the interval sequence of interval
order $P$ and $(b,e)\in \irep(P)$, and $\Sigma(P)$ for the set of all
interval sequences of interval representations of $P$.
\begin{lemma}\label{P:irep1}
  If $\sigma\in \Sigma(P)$ and $\sigma\approx^\ast \sigma'$, then
  $\sigma'\in \Sigma(P)$. 
\end{lemma}
\begin{proof}[sketch]
  We show that $\sigma\in \Sigma(P)$ and $\sigma\approx \sigma'$
  imply $\sigma'\in \Sigma(P)$.  Suppose that
  $\sigma=\sigma_1b(x)b(y)\sigma_2$ and
  $\sigma'= \sigma_1b(y)b(x)\sigma_2$ and that $(b,e)\in \irep(P)$
  generates $\sigma$. Then $(b',e)$ with
  \begin{equation*}
    b'(z) =
    \begin{cases}
      b(y), & \text{ if } z = x,\\
b(x), & \text{ if } z = y,\\
b(z), & \text{ otherwise}
    \end{cases}
  \end{equation*}
  is in $\irep(P)$, as $b'(x) <_Q e(x)$, $b'(y) <_Q e(y)$ and,
  for all $v,w\in P$, $v <_P w \Leftrightarrow e(v) <_P b(w)$ still
  holds. In addition, $(b',e)$ generates $\sigma'$. An analogous
  result for $\sigma=\sigma_1e(x)e(y)\sigma_2$ and
  $\sigma'= \sigma_1e(y)e(x)\sigma_2$ holds by opposition. The result
  for $\approx^\ast$ follows by a simple induction. \qed
\end{proof}
\begin{lemma}\label{P:irep2}
  Let $P$ be an interval order. If $(b,e),(b',e')\in \irep(P)$ assign $b$ and $e$ to elements of $P$ in
  interval sequences, then
  $\sigma_{(b,e)}(P) \approx^\ast \sigma_{(b',e')}(P)$.
\end{lemma}
\begin{proof}[sketch]
  Let $\prec_1$ and $\prec_2$ be the orderings of the interval
  sequences for $(b,e)$ and $(b',e')$ in $Q$. Then $b(x)\prec_1 e(x)$
  and $b(x)\prec_2 e(x)$ for all $x\in X$, and
  $e(x)\prec_1 b(y) \Leftrightarrow e(x)\prec_2 b(y)$ for all
  $x,y\in X$. It follows that there is no $b(z)$ in $\prec_1$ or
  $\prec_2$ between the positions of $e(x)$ in $\prec_1$ and $\prec_2$
  and, by opposition, there is no $e(z)$ in $\prec_1$ or $\prec_2$
  between the positions of $b(x)$ in $\prec_1$ and $\prec_2$. But this
  means that the positions of $e(x)$ and $b(x)$ can be rearranged by
  $\approx^\ast$. \qed
\end{proof}
\begin{proposition}\label{P:irep}
  If $P$ is an interval order and $(b,e)\in\irep(P)$, then
  $[\sigma_{(b,e)}(P)]_{\approx^\ast} = \Sigma(P)$.  The mapping
  $\varphi$ defined by $\varphi( P)= [\sigma_{(b,e)}(P)]_{\approx^\ast}$
  is a bijection.
\end{proposition}
\begin{proof}
  By Lemma~\ref{P:irep1} and \ref{P:irep2}, and properties of
  interval representations.\qed
\end{proof}

\section{Posets with interfaces}
\label{se:iposets}

An element $s$ of  poset $( P, \le)$ is \emph{minimal}
(\emph{maximal}) if $v\not\le s$ ($v\not\ge s$) holds for all
$v\in P$.  We write $P_{\min}$ ($P_{\max}$) for the sets of  minimal
(maximal) elements of $P$.

\begin{definition}
  A \emph{poset with interfaces (iposet)} consists of a poset $P$
  together with two injective morphisms 
  \begin{equation*}
    \label{eq:ipos}
    \xymatrix@C=1.5pc@R=0pc{%
      [ n] \ar[dr]^s && [ m]
      \ar[dl]_t \\
      & P &
    }
\end{equation*}
%  $s:[ n]\to P\from[ m]: t$ 
such that $s[ n]\subseteq P_{\min}$ and $t[ m]\subseteq P_{\max}$.
\end{definition}

Injection $s:[ n]\to P$ represents the \emph{source interface} of $P$
and $t:[ m]\to P$ its \emph{target interface}.  We write
$(s,P,t):n\to m$ for the iposet $s:[ n]\to P\from[ m]: t$.

Figure~\ref{fi:iposets} shows some examples of iposets. Elements of
source and target interfaces are depicted as filled half-circles to
indicate the unfinished nature of the events they represent.

\begin{figure}[tp]
\centering
  \begin{tikzpicture}[->, >=latex', x=1.2cm, y=.8cm, label
    distance=-.25cm, shorten <=-3pt, shorten >=-3pt]
    \begin{scope}
      \begin{scope}
        \path[use as bounding box] (0,0) to (1,-1);
        \node (1) at (0,0) {\intpt};
        \node (2) at (1,0) {\intpt};
        \node (3) at (0,-1) {\intpt};
        \node (4) at (1,-1) {\intpt};
        \path (1) edge (2);
        \path (3) edge (2);
        \path (3) edge (4);
      \end{scope}
      \begin{scope}[xshift=3cm]
        \path[use as bounding box] (0,0) to (1,-1);
        \node [label=left:{\tiny $1$}] (1) at (0,0) {\inpt};
        \node (2) at (1,0) {\intpt};
        \node (3) at (0,-1) {\intpt};
        \node (4) at (1,-1) {\intpt};
        \path (1) edge (2);
        \path (3) edge (2);
        \path (3) edge (4);
        %\node (l1) at (-.3,0) {\small $1\;\;$};
        %\path[-, densely dashed] (l1.center) edge (1.center);
      \end{scope}
      \begin{scope}[xshift=6cm]
        \path[use as bounding box] (0,0) to (1,-1);
        \node (1) at (0,0) {\intpt};
        \node (2) at (1,0) {\intpt};
        \node [label=left:{\tiny 1}] (3) at (0,-1) {\inpt};
        \node (4) at (1,-1) {\intpt};
        \path (1) edge (2);
        \path (3) edge (2);
        \path (3) edge (4);
        % \node (l1) at (-.3,-1) {\small $1\;\;$};
        % \path[-, densely dashed] (l1.center) edge (3.center);
      \end{scope}
      \begin{scope}[xshift=9cm]
        \path[use as bounding box] (0,0) to (1,-1);
        \node [label=left:{\tiny 1}] (1) at (0,0) {\inpt};
        \node (2) at (1,0) {\intpt};
        \node [label=left:{\tiny 2}] (3) at (0,-1) {\inpt};
        \node (4) at (1,-1) {\intpt};
        \path (1) edge (2);
        \path (3) edge (2);
        \path (3) edge (4);
        % \node (l1) at (-.3,0) {\small $1\;\;$};
        % \node (l2) at (-.3,-1) {\small $2\;\;$};
        % \path[-, densely dashed] (l1.center) edge (1.center);
        % \path[-, densely dashed] (l2.center) edge (3.center);
      \end{scope}
    \end{scope}
    \begin{scope}[yshift=-1.7cm]
      \begin{scope}
        \path[use as bounding box] (0,0) to (1,-1);
        \node (1) at (0,0) {\intpt};
        \node (2) at (1,0) {\intpt};
        \node (3) at (0,-1) {\intpt};
        \node [label=right:{\tiny 1}] (4) at (1,-1) {\outpt};
        \path (1) edge (2);
        \path (3) edge (2);
        \path (3) edge (4);
        % \node (r1) at (1.3,-1) {\small $\;\;1$};
        % \path[-, densely dashed] (r1.center) edge (4.center);
      \end{scope}
      \begin{scope}[xshift=3cm]
        \path[use as bounding box] (0,0) to (1,-1);
        \node [label=left:{\tiny 1}] (1) at (0,0) {\inpt};
        \node [label=right:{\tiny 1}] (2) at (1,0) {\outpt};
        \node (3) at (0,-1) {\intpt};
        \node (4) at (1,-1) {\intpt};
        \path (1) edge (2);
        \path (3) edge (2);
        \path (3) edge (4);
        % \node (l1) at (-.3,0) {\small $1\;\;$};
        % \path[-, densely dashed] (l1.center) edge (1.center);
        % \node (r1) at (1.3,0) {\small $\;\;1$};
        % \path[-, densely dashed] (r1.center) edge (2.center);
      \end{scope}
      \begin{scope}[xshift=6cm]
        \path[use as bounding box] (0,0) to (1,-1);
        \node [label=left:{\tiny 1}]  (1) at (0,0) {\inpt};
        \node [label=right:{\tiny 1}] (2) at (1,0) {\outpt};
        \node [label=left:{\tiny 2}] (3) at (0,-1) {\inpt};
        \node [label=right:{\tiny 2}] (4) at (1,-1) {\outpt};
        \path (1) edge (2);
        \path (3) edge (2);
        \path (3) edge (4);
        % \node (l1) at (-.3,-1) {\small $1\;\;$};
        % \path[-, densely dashed] (l1.center) edge (3.center);
        % \node (r1) at (1.3,0) {\small $\;\;1$};
        % \node (r2) at (1.3,-1) {\small $\;\;2$};
        % \path[-, densely dashed] (r1.center) edge (2.center);
        % \path[-, densely dashed] (r2.center) edge (4.center);
      \end{scope}
      \begin{scope}[xshift=9cm]
        \path[use as bounding box] (0,0) to (1,-1);
        \node [label=left:{\tiny 1}]  (1) at (0,0) {\inpt};
        \node [label=right:{\tiny 2}] (2) at (1,0) {\outpt};
        \node [label=left:{\tiny 2}] (3) at (0,-1) {\inpt};
        \node [label=right:{\tiny 1}] (4) at (1,-1) {\outpt};
        \path (1) edge (2);
        \path (3) edge (2);
        \path (3) edge (4);
        % \node (l1) at (-.3,0) {\small $1\;\;$};
        % \node (l2) at (-.3,-1) {\small $2\;\;$};
        % \path[-, densely dashed] (l1.center) edge (1.center);
        % \path[-, densely dashed] (l2.center) edge (3.center);
        % \node (r1) at (1.3,0) {\small $\;\;1$};
        % \node (r2) at (1.3,-1) {\small $\;\;2$};
        % \path[-, densely dashed] (r1.center) edge (2.center);
        % \path[-, densely dashed] (r2.center) edge (4.center);
      \end{scope}
    \end{scope}
  \end{tikzpicture}%
  % \hspace*{2cm}
  \caption{Eight of 25 different iposets based on poset \N.}
  \label{fi:iposets}
\end{figure}
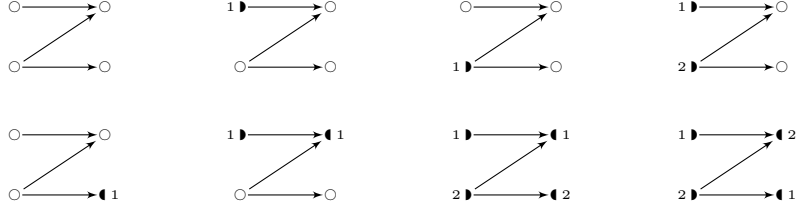

Next we define a sequential gluing composition on iposets whose
interfaces agree and we adapt the standard parallel composition of
posets to iposets.
\begin{definition}
  Let $(s_1,P_1,t_1):n\to m$ and
  $(s_2,P_2 ,t_2):\ell\to k$ be iposets.
  \begin{enumerate}[nosep]
  \item For $m=\ell$, their \emph{gluing composition} is the iposet
    $(s_1,P_1\ipomconcat P_2,t_2):n\to k$ with
    $P_1\ipomconcat P_2 = \left(( P_1\sqcup P_2)_{/ t_1( i)= s_2(
        i)},\mathord\le_1\cup \mathord\le_2\cup( P_1\setminus t_1[
      m])\times( P_2\setminus s_2[ m])\right)$.
  \item Their \emph{parallel composition} is the iposet
    $(s,P_1\otimes P_2,t):n+\ell\to m+k$ with
    $s=\left(s_1\otimes s_2\right)\circ \varphi_{ n, l}$ and
    $t=\left(t_1\otimes t_2\right)\circ \varphi_{ m, k}$.
  \end{enumerate}
\end{definition}

Parallel composition of iposets thus puts components ``side by side'':
it is the disjoint union of posets and interfaces. Gluing composition
puts iposets ``one after the other'', $P_1$ before $P_2$, but glues
their interfaces together (and adds arrows from all points in $P_1$
that are not in its target interface to all points in $P_2$ that are
not in its source interface).  Figures~\ref{fi:ndecomp}
and~\ref{fi:iposetcomp} show examples. The half-circles in source and
target interfaces are glued to circles in the diagrams.

\begin{figure}[bp]
  \centering
  \begin{tikzpicture}[->, >=latex', x=1.3cm, y=.8cm, label distance=-.25cm, shorten <=-3pt, shorten >=-3pt]
    \begin{scope}
      \node (1) at (.5,0) {\intpt};
      \node (2) at (0,-1) {\intpt};
      \node [label=right:{\tiny 1}]  (3) at (1,-1) {\outpt};
      \path (2) edge (3);
      % \node (r1) at (1.2,-1) {\small $\;\;1$};
      % \path[-, densely dashed] (r1.center) edge (3.center);
      \node at (1.5,-.5) {$\ipomconcat$};
    \end{scope}
    \begin{scope}[shift={(2,0)}]
      \node (1) at (0,0) {\intpt};
      \node [label=left:{\tiny 1}]  (2) at (0,-1) {\inpt}; 
      % \node (l1) at (-.3,-1) {\small $1\;\;$};
      % \path[-, densely dashed] (l1.center) edge (2.center);
      \node at (.5,-.5) {$=$};
    \end{scope}
   \begin{scope}[shift={(3,0)}]
      \node (1) at (0,0) {\intpt};
      \node (2) at (1,0) {\intpt};
      \node (3) at (0,-1) {\intpt};
      \node (4) at (1,-1) {\intpt};
      \path (1) edge (2);
      \path (3) edge (2);
      \path (3) edge (4);
      \node at (1.5,-.5) {$=$};
    \end{scope}
    \begin{scope}[shift={(5,0)}]
      \node  [label=right:{\tiny 1}] (1) at (0,0) {\outpt};
      \node (2) at (0,-1) {\intpt}; 
      % \node (r1) at (.3,0) {\small $\;\;1$};
      % \path[-, densely dashed] (r1.center) edge (1.center);
      \node at (.5,-.5) {$\ipomconcat$};
    \end{scope}
    \begin{scope}[shift={(6,0)}]
      \node  [label=left:{\tiny 1}] (1) at (0,0) {\inpt};
      \node (2) at (1,0) {\intpt};
      \node (3) at (.5,-1) {\intpt};
      \path (1) edge (2);
    \end{scope}
  \end{tikzpicture}
  \caption{Two different decompositions of the \N.}
  \label{fi:ndecomp}
\end{figure}
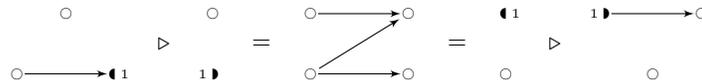

\begin{figure}[t]
  \centering
  \begin{tikzpicture}[->, >=latex', x=1.3cm, y=.8cm, label
    distance=-.25cm,shorten <=-3pt, shorten >=-3pt]
    \begin{scope}
      \begin{scope}
        \begin{scope}
          \path[use as bounding box] (0,0) to (1,-1);
          \node (1) at (0,0) {\intpt};
          \node (2) at (1,0) {\intpt};
          \node (3) at (0,-1) {\intpt};
          \node (4) at (1,-1) {\intpt};
          \path (1) edge (2);
          \path (3) edge (2);
          \path (3) edge (4);
          \node at (1.5,-.5) {$\ipomconcat$};
        \end{scope}
        \begin{scope}[shift={(2,0)}]
          \path[use as bounding box] (0,0) to (1,-1);
          \node (1) at (0,0) {\intpt};
          \node (2) at (1,0) {\intpt};
          \node (3) at (0,-1) {\intpt};
          \node (4) at (1,-1) {\intpt};
          \path (1) edge (2);
          \path (3) edge (2);
          \path (3) edge (4);
          \node at (1.5,-.5) {$=$};
        \end{scope}
        \begin{scope}[shift={(4,0)}]
          \node (1) at (0,0) {\intpt};
          \node (2) at (1,0) {\intpt};
          \node (3) at (2,0) {\intpt};
          \node (4) at (3,0) {\intpt};
          \node (5) at (0,-1) {\intpt};
          \node (6) at (1,-1) {\intpt};
          \node (7) at (2,-1) {\intpt};
          \node (8) at (3,-1) {\intpt};
          \foreach \i/\j in {1/2,2/3,2/7,3/4,5/2,5/6,6/3,6/7,7/4,7/8}
          \path (\i) edge (\j);
        \end{scope}
      \end{scope}
      \begin{scope}[shift={(0,-2)}]
        \begin{scope}
          \path[use as bounding box] (0,0) to (1,-1);
          \node (1) at (0,0) {\intpt};
          \node [label=right:{\tiny 1}] (2) at (1,0) {\outpt};
          \node (3) at (0,-1) {\intpt};
          \node (4) at (1,-1) {\intpt};
          \path (1) edge (2);
          \path (3) edge (2);
          \path (3) edge (4);
          % \node (r1) at (1.2,0) {\small $\;\;1$};
          % \path[-, densely dashed] (r1.center) edge (2.center);
          \node at (1.5,-.5) {$\ipomconcat$};
        \end{scope}
        \begin{scope}[shift={(2,0)}]
          \path[use as bounding box] (0,0) to (1,-1);
          \node [label=left:{\tiny 1}] (1) at (0,0) {\inpt};
          \node (2) at (1,0) {\intpt};
          \node (3) at (0,-1) {\intpt};
          \node (4) at (1,-1) {\intpt};
          \path (1) edge (2);
          \path (3) edge (2);
          \path (3) edge (4);
          % \node (l1) at (-.2,0) {\small $1\;\;$};
          % \path[-, densely dashed] (l1.center) edge (1.center);
          \node at (1.5,-.5) {$=$};
        \end{scope}
        \begin{scope}[shift={(4,0)}]
          \node (1) at (0,0) {\intpt};
          \node (23) at (1.5,0) {\intpt};
          \node (4) at (3,0) {\intpt};
          \node (5) at (0,-1) {\intpt};
          \node (6) at (1,-1) {\intpt};
          \node (7) at (2,-1) {\intpt};
          \node (8) at (3,-1) {\intpt};
          \foreach \i/\j in {1/23,1/7,23/4,5/23,5/6,6/7,7/4,7/8}
          \path (\i) edge (\j);
        \end{scope}
      \end{scope}
      \begin{scope}[shift={(0,-4)}]
        \begin{scope}
          \path[use as bounding box] (0,0) to (1,-1);
          \node (1) at (0,0) {\intpt};
          \node [label=right:{\tiny 1}] (2) at (1,0) {\outpt};
          \node (3) at (0,-1) {\intpt};
          \node (4) at (1,-1) {\intpt};
          \path (1) edge (2);
          \path (3) edge (2);
          \path (3) edge (4);
          % \node (r1) at (1.2,0) {\small $\;\;1$};
          % \path[-, densely dashed] (r1.center) edge (2.center);
          \node at (1.5,-.5) {$\ipomconcat$};
        \end{scope}
        \begin{scope}[shift={(2,0)}]
          \path[use as bounding box] (0,0) to (1,-1);
          \node (1) at (0,0) {\intpt};
          \node (2) at (1,0) {\intpt};
          \node  [label=left:{\tiny 1}] (3) at (0,-1) {\inpt};
          \node (4) at (1,-1) {\intpt};
          \path (1) edge (2);
          \path (3) edge (2);
          \path (3) edge (4);
          % \node (l1) at (-.2,-1) {\small $1\;\;$};
          % \path[-, densely dashed] (l1.center) edge (3.center);
          \node at (1.5,-.5) {$=$};
        \end{scope}
        \begin{scope}[shift={(4,0)}]
          \node (1) at (0,0) {\intpt};
          \node (23) at (1.5,0) {\intpt};
          \node (3) at (2,0) {};
          \node (4) at (3,0) {\intpt};
          \node (5) at (0,-1) {\intpt};
          \node (6) at (1,-1) {\intpt};
          \node (7) at (2,-1) {\intpt};
          \node (8) at (3,-1) {\intpt};
          \foreach \i/\j in {1/23,1/7,23/4,5/23,5/6,6/4,6/7,23/8,7/8}
          \path (\i) edge (\j);
        \end{scope}
      \end{scope}
      \begin{scope}[shift={(0,-6)}]
        \begin{scope}
          \path[use as bounding box] (0,0) to (1,-1);
          \node (1) at (0,0) {\intpt};
          \node  [label=right:{\tiny 1}] (2) at (1,0) {\outpt};
          \node (3) at (0,-1) {\intpt};
          \node  [label=right:{\tiny 2}] (4) at (1,-1) {\outpt};
          \path (1) edge (2);
          \path (3) edge (2);
          \path (3) edge (4);
          %\node (r1) at (1.2,0) {\small $\;\;1$};
          % \path[-, densely dashed] (r1.center) edge (2.center);
          % \node (r2) at (1.2,-1) {\small $\;\;2$};
          % \path[-, densely dashed] (r2.center) edge (4.center);
          \node at (1.5,-.5) {$\ipomconcat$};
        \end{scope}
        \begin{scope}[shift={(2,0)}]
          \path[use as bounding box] (0,0) to (1,-1);
          \node  [label=left:{\tiny 1}]  (1) at (0,0) {\inpt};
          \node (2) at (1,0) {\intpt};
          \node  [label=left:{\tiny 2}] (3) at (0,-1) {\inpt};
          \node (4) at (1,-1) {\intpt};
          \path (1) edge (2);
          \path (3) edge (2);
          \path (3) edge (4);
          % \node (l1) at (-.2,0) {\small $1\;\;$};
          % \path[-, densely dashed] (l1.center) edge (1.center);
          % \node (l2) at (-.2,-1) {\small $2\;\;$};
          % \path[-, densely dashed] (l2.center) edge (3.center);
          \node at (1.5,-.5) {$=$};
        \end{scope}
        \begin{scope}[shift={(4,0)}]
          \node (1) at (0,0) {\intpt};
          \node (23) at (1.5,0) {\intpt};
          \node (4) at (3,0) {\intpt};
          \node (5) at (0,-1) {\intpt};
          \node (67) at (1.5,-1) {\intpt};
          \node (8) at (3,-1) {\intpt};
          \foreach \i/\j in {1/23,23/4,5/67,5/23,67/8,67/4,1/8}
          \path (\i) edge (\j);
        \end{scope}
      \end{scope}
    \end{scope}
  \end{tikzpicture}
  \caption{Four gluings of different \N s with interfaces.}
  \label{fi:iposetcomp}
\end{figure}
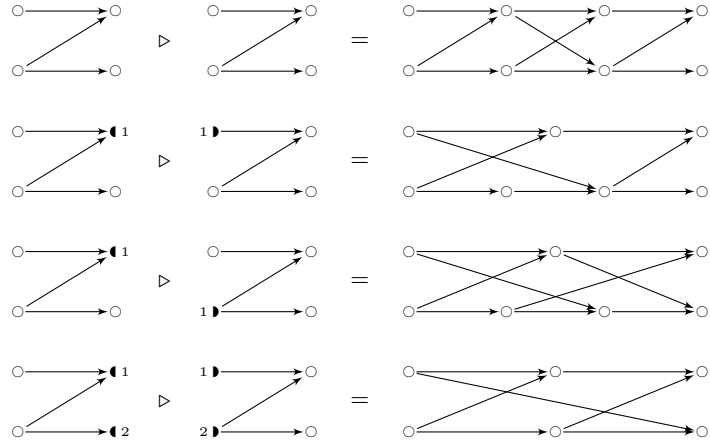

We define \emph{identity iposets} $\id_n=( \id,[ n], \id): n\to n$,
for $n\ge 0$. For convenience, we generalise this notation to other
singleton posets with interfaces: for $k, \ell\le n$, we write
$\idpos k \ell n$ for the iposet $(f_k^n,[n],f_\ell^n):k\to \ell$,
where $f_k^n:[ k]\to[ n]$ is the (identity) injection $x\mapsto x$
(similarly for $f_\ell^n$). Hence $\id_n= \idpos n n n$. We write
$\mcal S = \{ \idpos k \ell 1\mid k, \ell= 0, 1\}$ for the set of
all singleton iposets.

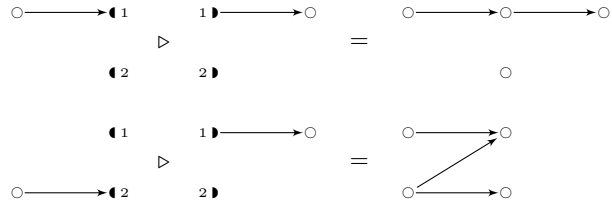
\begin{figure}[b]
  \centering
  \begin{tikzpicture}[->, >=latex', x=1.3cm, y=.8cm, label
    distance=-.25cm,shorten <=-3pt, shorten >=-3pt]
    \begin{scope}
      \begin{scope}
        \node (1) at (0,0) {\intpt};
        \node [label=right:{\tiny 1}] (2) at (1,0) {\outpt};
        \node [label=right:{\tiny 2}] (3) at (1,-1) {\outpt};
        \path (1) edge (2);
        \node at (1.5,-.5) {$\ipomconcat$};
      \end{scope}
      \begin{scope}[shift={(2,0)}]
        \node [label=left:{\tiny 1}] (1) at (0,0) {\inpt};
        \node (2) at (1,0) {\intpt};
        \node [label=left:{\tiny 2}] (3) at (0,-1) {\inpt};
        \path (1) edge (2);
        \node at (1.5,-.5) {$=$};
      \end{scope}
      \begin{scope}[shift={(4,0)}]
        \node (1) at (0,0) {\intpt};
        \node (2) at (1,0) {\intpt};
        \node (3) at (2,0) {\intpt};
        \node (4) at (1,-1) {\intpt};
        \path (1) edge (2);
        \path (2) edge (3);
      \end{scope}
    \end{scope}
    \begin{scope}[shift={(0,-2)}]
      \begin{scope}
        \node [label=right:{\tiny 1}] (3) at (1,0) {\outpt};
        \node (1) at (0,-1) {\intpt};
        \node [label=right:{\tiny 2}] (2) at (1,-1) {\outpt};
        \path (1) edge (2);
        \node at (1.5,-.5) {$\ipomconcat$};
      \end{scope}
      \begin{scope}[shift={(2,0)}]
        \node [label=left:{\tiny 1}] (1) at (0,0) {\inpt};
        \node (2) at (1,0) {\intpt};
        \node [label=left:{\tiny 2}] (3) at (0,-1) {\inpt};
        \path (1) edge (2);
        \node at (1.5,-.5) {$=$};
      \end{scope}
      \begin{scope}[shift={(4,0)}]
        \node (1) at (0,0) {\intpt};
        \node (2) at (1,0) {\intpt};
        \node (3) at (0,-1) {\intpt};
        \node (4) at (1,-1) {\intpt};
        \path (1) edge (2);
        \path (3) edge (4);
        \path (3) edge (2);
      \end{scope}
    \end{scope}
  \end{tikzpicture}
  \caption{Non-isomorphic gluings of symmetric parallel compositions.}
  \label{fi:otimesnotsymm}
\end{figure}

Parallel composition need not be commutative, as the namings of
interfaces in $P\otimes Q$ may differ from those in $Q\otimes P$.  One
can, however, rename interfaces using \emph{symmetries}: iposets
$(s,[ n],t):n\to n$ with $s$ and $t$ bijective.
Figure~\ref{fi:otimesnotsymm} shows two 
parallel compositions where renaming of interfaces and gluing with
another iposet yields  non-isomorphic posets.
% Categorically, the symmetries on $[n]$ form the symmetric group
% $S_n$ and identify iposets up-to renaming of interfaces.

Also, gluing and parallel composition need not satisfy an interchange
law:
\begin{equation*}
  (\idpos 0 0 1  \otimes \idpos 0 0 1) \ipomconcat (\idpos 0 0 1 \otimes \idpos 0 0 1) 
=   \pomset{\cdot \ar[r] \ar[dr] & \cdot \\ \cdot\ar[r] \ar[ur] &
  \cdot} 
\neq   \pomset{\cdot \ar[r] & \cdot \\ \cdot\ar[r] & \cdot} 
=(\idpos 0 0 1  \ipomconcat \idpos 0 0 1) \otimes (\idpos 0 0 1\ipomconcat \idpos 0 0 1)\,.
\end{equation*}
Hence iposets do \emph{not} form  (strict) monoidal categories, or even
PROPs, because $\otimes$ is not a tensor.
% ---although it is symmetric.
The situation differs from gluing compositions where interfaces of
iposets are defined by \emph{all} minimal and maximal
elements~\cite{DBLP:journals/ipl/Winkowski77}, and also from
sequential compositions of digraphs with ``partial'' interfaces
similar to ours where interface points glue arrows together and
disappear in these compositions~\cite{DBLP:conf/birthday/FioreC13}.
Both of these give rise to a PROP.

Gluing composition, of course, is not commutative either:
\begin{equation*}
\idpos 0 1 1\ipomconcat \idpos 1 0 1= \idpos 0 0 1 = \pomset{\cdot }\neq 
\pomset{\cdot \ar[r] & \cdot}=\idpos 1 0 1\ipomconcat \idpos 0 1 1
\end{equation*}

\begin{proposition}
  \label{pr:ipos-comp}
  Iposets form a small category with natural numbers as objects,
  iposets $( s, P, t): n\to m$ as morphisms, $\ipomconcat$ as composition,
  and identities $\id_n$.
\end{proposition}

Checking associativity of $\ipomconcat$ and the existence of units is
routine, as is the proof of the next proposition.

\begin{proposition}
  \label{pr:ops-assoc}
  Iposets form a monoid with composition $\otimes$ and unit
  $\id_0$.
\end{proposition}

% Another way to understand iposets is as objects of an arrow category.
% Here, the category of iposets is posed as a subcategory of the double
% arrow category $\fSet\to \Pos\from \fSet$ spanned by iposets.  (Here
% we use the fact that the full subcategory of $\Pos$ spanned by
% discrete posets $[ n]$ is isomorphic to the category $\fSet$ of finite
% sets and functions.)

A \emph{morphism} of iposets is a commuting diagram
\begin{equation}
  \label{eq:morphipos}
  \vcenter{
    \xymatrix{%
      [ n] \ar[r]^s \ar[d]_\nu & P \ar[d]_f & [ m]
      \ar[l]_t \ar[d]^\mu
      \\
      [ n'] \ar[r]_{ s'} & P' & [ m'] \ar[l]^{ t'}
    }}
\end{equation}
where $\nu$ and $\mu$ are strictly order preserving with respect to
$<_\Nat$ and $f$ is an order morphism.  Intuitively, iposet morphisms
thus preserve interfaces and their order in $\Nat$.  Let $\iPos$
denote the so-defined category.

An iposet morphism $(\nu,f,\mu)$ is an \emph{isomorphism} if $\nu$,
$f$ and $\mu$ are order isomorphisms.  Hence $n=n'$, $m=m'$,
$\nu=\id:n\to n$, and $\mu=\id:m\to m$ in
diagram~\eqref{eq:morphipos}.  As a consequence, we note that iposets
which are related by a symmetry $(s,[ n],t):n\to n$ need not be
isomorphic.
% \begin{equation*}
%   \xymatrix{%
%     [ n] \ar[r]^s \ar[d]_\id & P \ar[d]_f & [ m]
%     \ar[l]_t \ar[d]^\id
%     \\
%     [ n] \ar[r]_{ s'} & P' & [ m] \ar[l]^{ t'}
%   }
%   \end{equation*}

We write $P\cong Q$ if there exists an isomorphism $\varphi:P\to Q$.
The following lemma shows that the two compositions respect
isomorphism.
% justifies that we consider iposets up-to isomorphism.

\begin{lemma}
  Let $P, P', Q, Q'$ be iposets. Then $P\cong P'$ and $Q\cong Q'$ imply
  $P\otimes Q\cong P'\otimes Q'$ and $P\ipomconcat Q\cong P'\ipomconcat Q'$.
\end{lemma}

\begin{proof}
  Let $\varphi: P\to P'$ and $\psi: Q\to Q'$ be (the poset components
  of) isomorphisms.  Define the functions
  $\varphi\otimes \psi: P\sqcup Q\to P'\sqcup Q'$ and
  $\varphi\ipomconcat \psi:( P\sqcup Q)_{/t_P(i)=s_Q(i)}\to (P'\sqcup Q')
  _{/t_{P'}(i)=s_{Q'}(i)}$ as
  \begin{equation*}
    (\varphi \mathop{\Box} \psi)( x)=
    \begin{cases}
      \varphi( x) &\text{if } x\in P\,, \\
      \psi( x) &\text{if } x\in Q\,,
    \end{cases}
  \end{equation*}
  for $\Box \in \{\otimes,\ipomconcat\}$. First, $\varphi\otimes\psi$ is
  obviously an isomorphism.  Second, $\varphi\ipomconcat\psi$ is
  well-defined because $\varphi\circ t_P( i)= \psi\circ s_Q( i)$ for
  all $i\in[ m]$, and easily seen to be an isomorphism as well. \qed
\end{proof}

We write $P\preceq Q$ if there is a bijective (on points)
morphism $\varphi:Q\to P$ between iposets $P$ and $Q$. Intuitively,
$P\preceq Q$ iff $P$ has more arrows and is therefore less parallel
than $Q$, while interfaces are preserved. Similar relations on posets
and pomsets, sometimes called \emph{subsumption}, are well
studied~\cite{DBLP:journals/fuin/Grabowski81,
  DBLP:journals/tcs/Gischer88}. In particular, $\preceq$ is a preorder
on (finite) iposets and a partial order up to isomorphism.

\begin{lemma}\label{P:lax-interchange}
  For iposets $P,P',Q,Q'$, the following lax interchange law
  holds:
\begin{equation*}
(P\otimes  P') \ipomconcat (Q\otimes Q') \preceq (P\ipomconcat Q)\otimes (P'\ipomconcat Q')
\end{equation*}
\end{lemma}
\begin{proof}
  Let $P_\ell = (P\otimes P') \ipomconcat (Q\otimes Q')$ and
  $P_r = (P\ipomconcat Q)\otimes (P'\ipomconcat Q')$.  First,
  $P_\ell= ( P\sqcup Q)_{/ t_P\equiv s_Q}\sqcup( P'\sqcup Q')_{/
    t_{P'}\equiv s_{Q'}}=( P\sqcup Q\sqcup P'\sqcup Q')_{ t_P\equiv
    s_Q, t_{P'}\equiv s_{Q'}} =P_r$,
  by definition of $\otimes$. Hence both posets have the same points,
  and we may choose $\varphi:P_r\to P_\ell$ to be the identity.  It
  remains to show that $\varphi$ is order preserving, which means that
  every arrow in $P_r$ must be in $P_\ell$. 

  Hence suppose $x\le_{P_r} y$, that is, $x\le_{P\ipomconcat Q} y$ or
  $x\le _{P'\ipomconcat Q'} y$.  In the first case, if $x\le_P y$ or
  $x\le_Q y$, then $x\le_{P\otimes P'} y$ or $x\le_{Q\otimes Q'} y$ and
  therefore $x\le_{P_\ell} y$; and if $x\in P\setminus t_P$ and
  $y\in Q\setminus s_Q$, then
  $x\in P\sqcup P'\setminus t_{ P\otimes P'}$ and
  $y\in Q\sqcup Q'\setminus s_{ Q\otimes Q'}$ and therefore
  $x\le_{P_\ell} y$, too. The second case is symmetric. Thus, in
  any case, $x\le_{P_\ell} y$. \qed
\end{proof}

In sum, the algebra of iposets is thus similar to concurrent
monoids~\cite{DBLP:journals/jlp/HoareMSW11}, but $\ipomconcat$ is a partial
operation with many units $\id_k$. As $\otimes$ is not a tensor, the
categorical structure of iposets is somewhat unusual and deserves
further exploration.

\begin{proposition}
  \label{pr_iposets_generalize}
  \label{pr:posiposadj}
  $\Pos$ embeds into $\iPos$ as iposets with both interfaces $[0]$,
  and likewise for morphisms.  The so-defined inclusion functor
  $J: \Pos\to \iPos$ is fully faithful and left adjoint to the
  forgetful functor $F: \iPos\to \Pos$ that maps $(s,P,t)$
  to $P$, hence $\Pos$ is coreflective in $\iPos$.  Under $F$, gluing
  composition of iposets becomes serial composition of posets, and
  parallel composition of iposets becomes that of posets (hence,
  commutative).
\end{proposition}

\begin{proof}
  It is clear that $J$ is a functor.  It is full because any morphism
  $\tilde f$ from $P: 0\to 0$ to $Q: 0\to 0$ in $\iPos$ must have the
  form $( \emptyset, f, \emptyset)= J f$ for some $f$ in $\Pos$.  It
  is faithful because
  $J f = ( \emptyset, f, \emptyset)=( \emptyset, g, \emptyset)= J g$
  implies $f= g$.
    For $P\in \Pos$ and $\tilde Q\in \iPos$, $J$ induces a natural
  bijection $J: \Pos( P, F \tilde Q)\cong \iPos( J P, \tilde Q)$,
  hence $J$ and $F$ are indeed adjoint.  The last claims about the
  operations are clear. \qed
\end{proof}

%%%%%%%%%%%%%%%%%%%%%%%%%%%%%%%%%%%%%%%%%%%%%%%%%%%%%%%%%

\section{Further Properties of Iposets}

We now derive additional algebraic properties of iposets, before
turning to the set of iposets generated by gluing and parallel
composition from singleton iposets.

For an iposet $P$ with order relation $\le$
we write $\mathord{\para}= \mathord{\not\le}\cap
\mathord{\not\ge}$. Hence $x\para y$ iff $x$ and $y$ are unrelated and
therefore \emph{independent}. 
% (this last relation signifies \emph{independence} of events).  If no
% confusion can arise, we might abuse notation and also use $s_P= s[ n]$
% and $t_P= t[ m]$ for the images of the respective functions in~$P$.

In addition to the lax interchange in Lemma~\ref{P:lax-interchange},
we prove an equational interchange law that shows that the equational
theory of $\iPos$ as given by the bimonoidal laws in
Propositions~\ref{pr:ipos-comp} and~\ref{pr:ops-assoc} is not free.
The lemmas further below then show that this law  is the
\emph{only} non-trivial additional identity.

\begin{lemma}[Interchange]
  \label{le:interchange-eq}
  For all iposets $P$, $Q$ and $k, \ell\in\{ 0, 1\}$,
 \begin{equation*}
   ( \idpos k 1 1\otimes P)\ipomconcat( \idpos 1 \ell 1\otimes Q)= \idpos k
   \ell 1\otimes( P\ipomconcat Q)\,.
\end{equation*}
\end{lemma}

\begin{proof}[sketch]
  The interface between $\idpos k 1 1$ and $\idpos 1 \ell 1$ forces
  these iposets to be glued separately to the rest in the gluing
  composition
  $( \idpos k 1 1\otimes P)\ipomconcat( \idpos 1 \ell 1\otimes Q)$.
  % Hence this side of the equation
  % reduces to $\idpos k \ell 1 \otimes( P\ipomconcat Q)$ as well.
  \qed
\end{proof}

% \begin{proof}
%   The sets of points of both sides are\uli{Move proof to appendix.}
%   $S=\{ *\}\sqcup( P\sqcup Q)_{ t_P( i)= s_Q( i)}$, where $*$ is the
%   unique point of $\idpos k 1 1\ipomconcat \idpos 1 \ell 1= \idpos k \ell 1$.
%   Also the interfaces obviously agree.  Let $\le$ and $\le'$ be the
%   orders on the left-hand and right-hand sides, respectively, and let
%   $x, y\in S$. It remains to show that $x\le y$ iff $x\le' y$.  We can
%   assume $x\ne y$.  If $x= *$ or $y= *$, then $x\para y$ and
%   $x\para' y$, so we can also assume that $x, y\in S \setminus \{*\}$.

%   Suppose $x\le y$. Then $x\le_{ P\ipomconcat Q} y$.  If $x, y\in P$ or
%   $x, y\in Q$, then $x\le' y$ is obvious.  Otherwise, if
%   $x\in P\setminus t_P$ and $y\in Q\setminus s_Q$ then
%   $x\notin t_{ \idpos k 1 1\otimes P}$ and
%   $y\notin s_{ \idpos 1 \ell 1\otimes Q}$, and hence $x\le' y$ holds
%   again by definition of $\ipomconcat$.

%   Suppose $x\le' y$.  If $x, y\in \idpos k 1 1\otimes P$, then
%   $x, y\in P$, $x\le_P y$ and thus $x\le y$; the proof of $x\le y$
%   from $x, y\in \idpos 1 \ell 1\otimes Q$ is similar.  Otherwise, if
%   $x\in \idpos k 1 1\otimes P\setminus t_{ \idpos k 1 1\otimes P}$ and
%   $y\in \idpos 1 \ell 1\otimes Q\setminus s_{ \idpos 1 \ell 1\otimes
%     Q}$,
%   then $x\in P\setminus t_P$ and $y\in Q\setminus s_Q$, hence
%   $x\le_{ P\ipomconcat Q} y$ and once again $x\le y$. \qed
% \end{proof}

One the one hand, it follows that singleton iposets in $\mcal S$ do
not interfere with compositions.  On the other hand,
Lemma~\ref{le:interchange-eq} shows that decompositions need not be
unique. The next lemma shows a kind of converse: if an iposet can be
decomposed by $\ipomconcat$ and also by $\otimes$, then all but one of the
components must be in $\mcal S$.
Henceforth, let
$\mcal C_1= \left\{ P_1\otimes\dotsm\otimes P_n\bigmid P_1,\dotsc,
P_n\in\mcal S \right\}$ denote the set of
multisets-with-interfaces, that is, iposets with discrete order.

\begin{lemma}[Decomposition]
  \label{le:decomp}
  Let $P= P_1\otimes P_2= Q_1\ipomconcat Q_2$ such that $P_1\ne \id_0$,
  $P_2\ne \id_0$, and $Q_1\ne \idpos k n n$, $Q_2\ne \idpos n k n$ for
  any $k\le n$.  Then $P_1\in \mcal C_1$ or $P_2\in \mcal C_1$.
\end{lemma}

\begin{proof}
  % The preconditions imply that $P$ is disconnected (as a directed
  % graph), and that $Q_1\setminus t_{ Q_1}\ne \emptyset$ and
  % $Q_2\setminus s_{ Q_2}\ne \emptyset$.  Now if
  % $Q_1\setminus t_{ Q_1}$ contains more than one point,
  Suppose $P_1\notin \mcal C_1$
  and $P_2\notin \mcal C_1$.  Then $P$ contains a $\twotwo$: there are
  $w, x\in P_1$ and $y, z\in P_2$ for which $w<_P x$, $y<_P z$,
  $w\para_P y$, $w\para_P z$, $x\para_P y$, and $x\para_P z$.

  If $w,y\notin Q_2$, then $w,y\in Q_1\setminus t_{ Q_1}$.  As
  $Q_2\ne \idpos n k n$ for any $k\le n$, there must be an element
  $v\in Q_2\setminus s_{ Q_2}$.  But then $w\le_P v$ and $y\le_P v$,
  which yields arrows between $w\in P_1$ and $y\in P_2$
  that contradict $P=P_1 \otimes P_2$.  A dual argument rules out that
  $x,z\notin Q_1$.

  It follows that $w\in Q_2$ or $y\in Q_2$.  Assume, without loss of
  generality, that $w\in Q_2$. Then $x\in Q_2\setminus s_{Q_2}$
  because $w\le_{P_1} x$.  Now if also $y\in Q_2$, then by the same
  argument, $z\in Q_2\setminus s_{ Q_2}$.  Hence $Q_2$ contains two
  different points which are not in its starting interface; and as
  $Q_1\setminus t_{ Q_1}$ is non-empty, this again establishes a
  connection between $x\in P_1$ and $z\in P_2$ which cannot exist.
  Hence $y\notin Q_2$, but then $y\in Q_1\setminus t_{ Q_1}$, so that
  $y\le_P x$, which contradicts $x\para_P y$. \qed
\end{proof}

The next lemma generalises Levi's lemma for
words~\cite{journals/bcms/Levi44}.

\begin{lemma}[Levi property]
  \label{le:Levi}
  Let $P\mathop{\Box} Q= U\mathop{\Box} V$ for
  $\Box\in \{\ipomconcat,\otimes\}$.  Then there is an $R$ so that either
  $P= U\mathop{\Box }R$ and $R\mathop{\Box} Q= V$, or
    $U= P\mathop{\Box} R$ and $R\mathop{\Box} V= Q$.
\end{lemma}
\begin{proof}
  The proof for $\otimes$ is trivial: If $P\otimes Q=U\otimes V$,
  then this iposet is partitioned into three components according to
  $P\sqcup Q$ and $U\sqcup V$. If the decomposition of $U$ and $V$
  happens within $P$, then there is an $R$ such that $P=U\otimes R$
  and $R\otimes Q=V$. Otherwise, if it happens within $Q$, then there
  exists an $R$ such that $U=P\otimes R$ and $R\otimes V$. Finally, if
  $P=U$ and $Q=V$, there is nothing to show.  The proof for $\ipomconcat$ is
  similar, but more tedious due to  gluing. \qed
\end{proof}
It is instructive to find the two cases in the decomposition of \N\ in
Figure~\ref{fi:ndecomp}.

Levi's lemma is an interpolation property: every
$P\mathop{\Box} Q= U\mathop{\Box} V$ has a common
factorisation---either $U\mathop{\Box} R\mathop{\Box} Q$ or
$P\mathop{\Box} R\mathop{\Box} V$. Hence sequential and gluing
decompositions at top level are equal up-to associativity (and unit
laws). 

% If $\otimes$ is considered up-to symmetries, then Levi's lemma
% simplifies further for this case. Each of $P$ and $Q$ can then be
% split into two parts such that $U$ is formed of one part of $P$ and
% one of $Q$, and $V$ of the remaining parts of $P$ and $Q$.

% \ulilong{%
%   See~\cite{journals/tac/LawsonW17} for the categorical version of
%   this.  In short, we have an equidivisible category; but there is
%   probably more to be explored.}

The three lemmas in this section are helpful for characterising the
iposets generated by $\ipomconcat$ and $\otimes$ from singletons. This is
the subject of the next section.

%%%%%%%%%%%%%%%%%%%%%%%%%%%%%%%%%%%%%%%%%%%%%%%%%%%%%%

\section{Generating Iposets}
\label{se:generate}

Recall that $\mcal S$ is the set of \emph{singleton} iposets.  It
contains the four iposets $\idpos 0 0 1$, $\idpos 0 1 1$,
$\idpos 1 0 1$ and $\idpos 1 1 1$, that is,
\begin{equation*}
  [ 0]\to[ 1]\from [ 0]\,, \qquad [ 0]\to[ 1]\from[ 1]\,,
  \qquad [ 1]\to[ 1]\from [ 0]\,, \qquad [ 1]\to[ 1]\from[ 1]\,,
\end{equation*}
with mappings uniquely determined.  We are interested in the sets of
iposets generated from singletons using $\ipomconcat$ and $\otimes$.
%
% However, the singleton $\idpos 0 0 1$ is not a generator, because it
% can be generated from $\idpos 0 1 1$ and $\idpos 1 0 1$:
% $\idpos 0 0 1 = \idpos 0 1 1\ipomconcat \idpos 1 0 1$. Hence let
% $\mcal G = \mcal S\setminus \{\idpos 0 0 1\}$. 
Note that strictly speaking, $\idpos 0 0 1$ should not count as a
generator, because by Lemma~\ref{le:interchange-eq} it is equal to
$\idpos 0 1 1\ipomconcat \idpos 1 0 1$.

\begin{definition}
  The set of \emph{gluing-parallel} iposets (\emph{gp-iposets}) is the
  smallest set that contains the empty iposet $\id_0$ and the
  singleton iposets in
  % $\mcal G$
  $\mcal S$ and is closed under gluing and parallel composition.
\end{definition}

\begin{theorem}\label{P:ipos-free}
  The gp-iposets are generated freely by
  % $\mcal G$
  $\mcal S$ in the variety of algebras satisfying the equations of
  Propositions~\ref{pr:ipos-comp} and~\ref{pr:ops-assoc} and
  Lemma~\ref{le:interchange-eq}.
\end{theorem}

\begin{proof}[sketch]
  Suppose $(A,\ipomconcat,\otimes,(1_i)_{i\ge 0})$ is any algebra satisfying
  the equations of Propositions~\ref{pr:ipos-comp}
  and~\ref{pr:ops-assoc} and Lemma~\ref{le:interchange-eq} and let
  % $\varphi:\mcal G\to A$
  $\varphi:\mcal S\to A$ be any function.  We need to show that $\varphi$
  extends to a unique iposet morphism $\hat{\varphi}$.

  We can generate any $\id_n$ as a parallel composition of $\id_1$. We
  map $\hat{\varphi}(\id_i)\mapsto 1_i$ for any $i\ge 0$, and we map
  any other singleton
  % $p\in \mcal G$
  $p\in \mcal S$ as $\hat{\varphi}(p) = \varphi(p)$.  For complex
  iposets we proceed by induction on the number of elements, assuming
  that homomorphism laws hold for iposets with $n$ elements. 

  If the top composition of the size $n+1$ iposet is $\ipomconcat$, then we
  use Levi's lemma to factorise with respect to $\ipomconcat$ and use
  associativity of $\ipomconcat$ to establish the homomorphism property of
  $\hat{\varphi}$. For $\otimes$ we proceed likewise. Finally, if the
  top composition is ambiguous, then the decomposition lemma forces
  the configuration in which the interchange lemma can be applied,
  yielding a parallel composition of the same size.  Finally, this
  extension is unique, as it was forced by the construction. \qed

  % We can then proceed by induction on the elements of the iposet and
  % use the decomposition properties to show that $\hat{\varphi}$ is
  % fixed for $\ipomconcat$ and $otimes$, and indeed a morphism.

  % Let $t_1= t_2$ be an identity in the language of iposets with
  % operations $\otimes$ and $\ipomconcat$.  We proceed by recursion. If the
  % top operations in $t_1$ and $t_2$ are the same, then we use Levi's
  % lemma~\ref{le:Levi} to factorise both terms in the same components
  % and check for component-wise equality. It they are different, then
  % the Decomposition Lemma~\ref{le:decomp} forces the identity to be
  % of the form in the Interchange Lemma~\ref{le:interchange-eq}.
  % \qed
\end{proof}

Now
% , for convenience, however, we put $\idpos 0 0 1$ back into
% $\mcal G$ and
we %
define hierarchies of iposets generated from $\mcal S$.
% These differ from that generated by $\mcal G$ .
(If $\idpos 0 0 1$ were removed from $\mcal S$, the hierarchy would be
different
only for less than two alternations of $\ipomconcat$ and $\otimes$.%
)

For any $\mcal Q\subseteq \iPos$ and
$\mathop{\Box}\in\{\otimes,\ipomconcat\}$, let
\begin{align*}
  \mcal Q^{\mathop{\Box}} = \{ P_1\mathop{\Box}\dotsm\mathop{\Box}  P_n\mid n\in \Nat,
  P_1,\dotsc, P_n\in \mcal Q\}\ .
\end{align*}
Then define $\mcal C_0= \mcal D_0= \mcal S$ and, for all $n\in
\Nat$,
\begin{equation*}
  \mcal C_{ 2 n+ 1}= \mcal C_{ 2 n}^\otimes\,, \qquad \mcal D_{ 2 n+
    1}= \mcal D_{ 2 n}^\ipomconcat\,, \qquad \mcal C_{ 2 n+ 2}= \mcal C_{ 2
    n+ 1}^\ipomconcat\,, \qquad \mcal D_{ 2 n+ 2}= \mcal D_{ 2 n+ 1}^\otimes
\end{equation*}
(this agrees with the $\mcal C_1$ notation used earlier).  Finally,
let
\begin{equation*}
  \bar{ \mcal S}\defeq \bigcup_{ n\ge 0} \mcal C_n= \bigcup_{ n\ge 0}
  \mcal D_n
\end{equation*}
be the set of all iposets generated from $\mcal S$ by application of
$\otimes$ and $\ipomconcat$.

\begin{lemma}
  \label{le:trivial_hierarchy_inclusions}
  For all $n\in \Nat$,
  $\mcal C_n\cup \mcal D_n\subseteq \mcal C_{ n+ 1}\cap \mcal D_{ n+
    1}$.
  % \begin{enumerate}
  % \item $\mcal C_n\subseteq \mcal C_{ n+ 1}$ and
  %   $\mcal C_n\subseteq \mcal D_{ n+ 1}$; as well as
  % \item $\mcal D_n\subseteq \mcal D_{ n+ 1}$ and $\mcal D_n\subseteq \mcal C_{ n+ 1}$.
  % \end{enumerate}
\end{lemma}

\begin{proof}
  We need to check the inclusions
  $\mcal C_n\subseteq \mcal C_{ n+ 1}$,
  $\mcal D_n\subseteq \mcal D_{ n+ 1}$,
  $\mcal C_n\subseteq \mcal D_{ n+ 1}$ and
  $\mcal D_0\subseteq \mcal C_1$. The first two are trivial by
  construction, plus $\mcal C_n\subseteq \mcal D_{ n+ 1}$ and
  $\mcal D_n\subseteq \mcal C_{ n+ 1}$.  For the third one, note that
  $\mcal C_0\subseteq \mcal C_{0}^{\ipomconcat}= \mcal S^\ipomconcat= \mcal
  D_{0}^{\ipomconcat} = \mcal D_{1}$.
  Since $\mcal C_{n}$ is constructed from $\mcal C_{0}$ by the same
  alternations of $\otimes$ and $\ipomconcat$ as $\mcal D_{n+1}$ is
  constructed from $\mcal D_{1}$, the inclusion holds. The proof of
  the fourth inclusion is similar. \qed
\end{proof}

\begin{theorem}
  \label{th:iorder}
  An iposet is in $\mcal C_2$ iff it is an interval order.
\end{theorem}

\begin{proof}
  Suppose $P\ipomconcat Q\in \mcal C_2$. First it is clear that all elements
  of $\mcal C_1$ are interval orders, so we will be done once we can
  show that the gluing composition of two interval orders is an
  interval orders.  This is precisely the proof of
  Lemma~\ref{le:decomp}: if $P\ipomconcat Q$ contains a $\twotwo$, then so
  do $P$ or $Q$. Yet we also give a direct construction: Let
  $\sigma_P$ be the interval sequence for interval representation
  $(b_P,e_P)$ of $P:n\to m$ and $\sigma_Q$ the interval sequence for
  interval representation $(b_Q,e_Q)$ of $Q:m\to k$. Then concatenate
  $\sigma_P$ and $\sigma_Q$, rename $b_P$, $b_Q$ as $b$ and $e_P$,
  $e_Q$ as $e$, delete $e(t_P(i))$, $b(s_Q(i))$ and replace
  $e(t_Q(i))$ with $e(t_P(i))$ for each $i\in [m]$. This yields the
  interval sequence for interval representation $(b,e)$ of $P\ipomconcat Q$
  and $P\ipomconcat Q$ is therefore an interval order.
%
% Let
%   $Q=\{[ a_i, b_i]\mid i= 1,\dotsc, k\}$,
%   $R=\{[ c_i, d_i]\mid i= 1,\dotsc, \ell\}$ be interval
%   representations.  Let $M=\max\{ b_i\mid i= 1,\dotsc, k\}$ be the
%   right boundary of $P$ and put
%   $R= P\cup\{[ c_i+ M+ 1, d_i+ M+ 1]\mid i= 1,\dotsc, \ell\}$.  We
%   will turn $R$ into an interval representation of $P\ipomconcat Q$ by
%   connecting some interval endpoints:
%
  % Write $t:[ m]\to P$ for the targets of $P$ and $s:[ m]\to Q$ for the
  % sources of $Q$.  For $i\in[ m]$, write $t( i)=[ a_i, b_i]$ and
  % $s( i)=[ c_i, d_i]$.  Now let $R_0= R$ and then, for each
  % $i= 1,\dotsc, m$ in that order, define
  % \begin{equation*}
  %   R_i= R_{ i- 1}\setminus\{[ a_i, b_i],[ c_i+ M+ 1, d_i+ M+
  %   1]\}\cup\{[ a_i, d_i+ M+ 1]\}\,;
  % \end{equation*}
  % that is, connect the two corresponding intervals by inserting $[
  % b_i, c_i+ M+ 1]$.
  Figure~\ref{fi:intcomp} gives an example.
% It is
%   clear that the so-constructed $R_m$ is an interval representation of
%   $P\ipomconcat Q$.

  \begin{figure}[tbp]
    \centering
    \begin{tikzpicture}[label distance=-.2cm,shorten <=-3pt, shorten >=-3pt]
      \begin{scope}[->, >=latex', xscale=.4, yscale=.75]
        \begin{scope}
          \node (1) at (1.5,0) {$a$};
          \node [label=right:{\tiny 1}] (2) at (4.5,0) {$b$};
          \node (3) at (1.5,-1) {$c$};
          \node (4) at (4.5,-1) {$d$};
          \node [label=right:{\tiny 2}] (5) at (4.5,-2) {$e$};
          \foreach \i/\j in {1/2,3/2,3/4,3/5} \path (\i) edge (\j);
          \node at (6.5,-1) {$\ipomconcat$};
        \end{scope}
        \begin{scope}[xshift=7.5cm]
          \node [label=left:{\tiny 1}]  (6) at (.5,0) {$f$};
          \node (7) at (3.5,0) {$g$};
          \node [label=left:{\tiny 2}] (8) at (.5,-2) {$h$};
          \node (9) at (3.5,-2) {$i$};
          \foreach \i/\j in {6/7,8/7,8/9} \path (\i) edge (\j);
          \node at  (5.7,-1) {$=$};
        \end{scope}
        %\path[-, densely dashed] (2) edge (6);
        %\path[-, densely dashed] (5) edge (8);
        \begin{scope}[xshift=16cm]
          \node (1) at (-.5,0) {$a$};
          \node (2) at (4.5,0) {$bf$};
          \node (3) at (-.5,-1) {$c$};
          \node (4) at (2.5,-1) {$d$};
          \node (5) at (4.5,-2) {$eh$};
        \end{scope}
        \begin{scope}[xshift=20cm]
          \node (6) at (.5,0) {\phantom{$bf$}};
          \node (7) at (5.5,0) {$g$};
          \node (8) at (.5,-2) {\phantom{$eh$}};
          \node (9) at (5.5,-2) {$i$};
        \end{scope}
        \foreach \i/\j in {1/2,3/2,3/4,3/5} \path (\i) edge (\j);
        \foreach \i/\j in {6/7,8/7,8/9} \path (\i) edge (\j);
        \foreach \i/\j in {1/9,4/7,4/9} \path (\i) edge (\j);
      \end{scope}
      \begin{scope}[-, xscale=.4, yshift=-2.5cm, shorten <=-4.5pt, shorten >=-4.5pt]]
        \begin{scope}
          \node (1l) at (0,0) {{\tiny $|$}};
          \node (1r) at (3,0) {{\tiny $|$}};
          \node (2l) at (4,0) {{\tiny $|$}};
          \node (2r) at (5,0) {{\tiny $|$}};
          \node (3l) at (0,-1) {{\tiny $|$}};
          \node (3r) at (1,-1) {{\tiny $|$}};
          \node (4l) at (2,-1) {{\tiny $|$}};
          \node (4r) at (5,-1) {{\tiny $|$}};
          \node (5l) at (2,-2) {{\tiny $|$}};
          \node (5r) at (5,-2) {{\tiny $|$}};
          \path (1l) edge node[above] {\small $I(a)$}  (1r);
          \path (2l) edge node[above] {\small $I(b)$}  (2r);
          \path (3l) edge node[above] {\small $I(c)$}  (3r);
          \path (4l) edge node[above] {\small $I(d)$}  (4r);
          \path (5l) edge node[above] {\small $I(e)$}  (5r);
          %\foreach \i in {1,2,3,4,5} \path (\i l) edge node[above] {\small $\i$} (\i r);
          \node at (6,-1) {$\ipomconcat$};
        \end{scope}
        \begin{scope}[xshift=7cm]
          \node (6l) at (0,0) {{\tiny $|$}};
          \node (6r) at (3,0) {{\tiny $|$}};
          \node (7l) at (4,0) {{\tiny $|$}};
          \node (7r) at (5,0) {{\tiny $|$}};
          \node (8l) at (0,-2) {{\tiny $|$}};
          \node (8r) at (1,-2) {{\tiny $|$}};
          \node (9l) at (2,-2) {{\tiny $|$}};
          \node (9r) at (5,-2) {{\tiny $|$}};
          \path (6l) edge node[above] {\small $I(f)$}  (6r);
          \path (7l) edge node[above] {\small $I(g)$}  (7r);
          \path (8l) edge node[above] {\small $I(h)$}  (8r);
          \path (9l) edge node[above] {\small $I(i)$}  (9r);
          %\foreach \i in {6,7,8,9} \path (\i l) edge node[above]
          %{\small $\i$} (\i r);
          \node at  (6,-1) {$=$};
        \end{scope}
        %\path[densely dashed] (2r) edge (6l);
        %\path[densely dashed] (5r) edge (8l);
        \begin{scope}[xshift=14cm]
          \node (1l') at (0,0) {{\tiny $|$}};
          \node (1r') at (3,0) {{\tiny $|$}};
          \node (2/6l') at (4,0) {{\tiny $|$}};
          \node (3l') at (0,-1) {{\tiny $|$}};
          \node (3r') at (1,-1) {{\tiny $|$}};
          \node (4l') at (2,-1) {{\tiny $|$}};
          \node (4r') at (5,-1) {{\tiny $|$}};
          \node (5/8l') at (2,-2) {{\tiny $|$}};
        \end{scope}
        \begin{scope}[xshift=21cm]
          \node (2/6r') at (3,0) {{\tiny $|$}};
          \node (7l') at (4,0) {{\tiny $|$}};
          \node (7r') at (5,0) {{\tiny $|$}};
          \node (5/8r') at (1,-2) {{\tiny $|$}};
          \node (9l') at (2,-2) {{\tiny $|$}};
          \node (9r') at (5,-2) {{\tiny $|$}};
        \end{scope}
           \path (1l') edge node[above] {\small $I(a)$}  (1r');
\path (2/6l') edge node[above] {\small $I(bf)$}  (2/6r');
\path (3l') edge node[above] {\small $I(c)$}  (3r');
\path (4l') edge node[above] {\small $I(d)$}  (4r');
\path (5/8l') edge node[above] {\small $I(eh)$}  (5/8r');
\path (7l') edge node[above] {\small $I(g)$}  (7r');
\path (9l') edge node[above] {\small $I(i)$}  (9r');
        % \foreach \i in {1,2/6,3,4,5/8,7,9} \path (\i l') edge
        % node[above] {\small $\smash{\i}$} (\i r');
      \end{scope}
    \end{tikzpicture}
    \caption{Two interval orders and their concatenation: above as
      iposets, below using their interval representations.  (Labels
      added for convenience.)}
    \label{fi:intcomp}
  \end{figure}
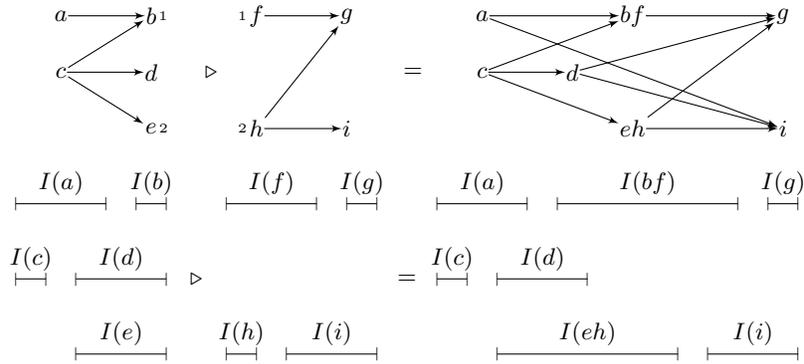
  
  For the backward direction, let $P$ be an interval order and $A_P$
  its set of maximal antichains.  Then $A_P$ is totally ordered by the
  relation $\sqsubset$ defined in Section~\ref{S:interval-orders}. 
  % \begin{equation*}
  %   A\preceq B\quad\text{iff}\quad \forall a\in A\setminus B: \forall
  %   b\in B\setminus A: a\le_P b\,,
  % \end{equation*}
  %see~\cite{DBLP:conf/apn/Janicki18, DBLP:journals/tcs/JanickiK93}.
  Now write $A_P=\{ P_1,\dotsc, P_k\}$ such that $P_i\sqsubset  P_j$ for
  $i< j$. Then each $P_i$ is an element of $\mcal S^\otimes$.  Write
  $s_1:[ n_1]\to P\from[ n_{ k+ 1}]: t_k$ for the sources and targets
  of $P$.

  For $i= 2,\dotsc, k$, let $[ n_i]= P_{ i- 1}\cap P_i$ be the overlap
  and $s_i:[ n_i]\hookrightarrow P_i$,
  $t_{ i- 1}:[ n_i]\hookrightarrow P_{ i- 1}$ the inclusions.
  Together with $s_1$ and $t_k$ this defines iposets
  $s_i:[ n_i]\to P_i\from[ n_{ i+ 1}]: t_i$.  (Note that
  $s_1:[ n_1]\to P_1$ because $P_1$ is the minimal element in $A_P$;
  similarly for $t_k:[ n_{ k+ 1}]\to P_k$.)  It is clear that
  $P= P_1\ipomconcat\dotsm\ipomconcat P_k$; see
  also~\cite[Prop.~2]{DBLP:conf/apn/Janicki18}. \qed
\end{proof}

In order to compare with series-parallel posets, we construct a
similar hierarchy for these.  Let
$\mcal T_0= \mcal U_0= \mcal S_0=\{ \idpos 0 0 1\}$ and, for all
$n\in \Nat$,
\begin{equation*}
  \mcal T_{ 2 n+ 1}= \mcal T_{ 2 n}^\otimes\,, \qquad \mcal U_{ 2 n+ 1}=
  \mcal U_{ 2 n}^\ipomconcat\,, \qquad \mcal T_{ 2 n+ 2}= \mcal T_{ 2 n+
    1}^\ipomconcat\,, \qquad \mcal U_{ 2 n+ 2}= \mcal U_{ 2 n+ 1}^\otimes\,.
\end{equation*}
Then, noting that any element of any $\mcal T_n$ or $\mcal U_n$ has
empty interfaces and that for iposets with empty interfaces, $\ipomconcat$
is serial composition, we see that
\begin{equation*}
  \bar{ \mcal S}_0\defeq \bigcup_{ n\ge 0} \mcal T_n= \bigcup_{ n\ge
    0} \mcal U_n
\end{equation*}
is the set of series-parallel posets.  Note that
$\mcal T_n\subseteq \mcal C_n$ and $\mcal U_n\subseteq \mcal D_n$ for
all $n$, hence also $\bar{ \mcal S}_0\subseteq \bar{ \mcal S}$.  Now
$\bar{ \mcal S}_0$ contains precisely the \N-free posets whereas \N\
is an interval order. Hence $\N\in \mcal C_2$, implying the next
lemma.  On the other hand, we will see below that
$\bar{ \mcal S}_0\not\subseteq \mcal C_n$ for any $n$.

\begin{lemma}
  $\mcal C_2\not\subseteq \bar{ \mcal S}_0$.
\end{lemma}

\begin{lemma}
  $\mcal C_1\cup \mcal D_1 \subsetneq \mcal C_{2} \cap \mcal D_{2}$,
  \ie~there is an iposet with two non-trivial different
  decompositions.
\end{lemma}

\begin{proof}
  Directly from Lemma~\ref{le:interchange-eq}.
 %  This follows from Lemma~\ref{le:interchange-eq}, but we also give a
 %  minimal example:\uli{Why? Remove example?}
 %  \begin{equation*}
 % \mcal C_1\cup \mcal
 %  D_1   \not\ni \pomset{\cdot \\ \cdot \ar[r] & \cdot} = \idpos 0 0 1\otimes( \idpos 0 0 1\ipomconcat \idpos 0 0 1)=( \idpos 0 1 1\otimes \idpos 0 1
 %  0)\ipomconcat( \idpos 1 0 1\otimes \idpos 0 1 1) \in
 %  \mcal C_{2} \cap \mcal D_{2}.
 %  \end{equation*}
\qed
\end{proof}

% First, it is obvious that
%   $P= \pomset{\cdot \\ \cdot \ar[r] & \cdot}\notin \mcal C_1\cup \mcal
%   D_1$. Seconly,  $P= \idpos 0 0 1\otimes( \idpos 0 0 1\ipomconcat \idpos 0 0 1)=( \idpos 0 1 1\otimes \idpos 0 1
%   0)\ipomconcat( \idpos 1 0 1\otimes \idpos 0 1 1)$ and therefore $P\in
%   \mcal C_{2} \cap \mcal D_{2}$. 
%   % \begin{equation*}
%   %   P= \binom{[ 0]\to[ 1]\from[ 0]}{[ 0]\to[ 1]\from
%   %     [ 0] \ipomconcat [ 0]\to[ 1]\from [ 1]},
%   % \end{equation*}
%   hence $P\in \mcal D_{2}$, and $P=( \idpos 0 1 1\otimes \idpos 0 1
%   0)\ipomconcat( \idpos 1 0 1\otimes \idpos 0 1 1)$,
%   % \begin{equation*}
%   %   P= \binom{[ 0]\to[ 1]\from[ 1]}{[ 0]\to[ 1]\from
%   %     [ 0]} \ipomconcat \binom{[ 1]\to[ 1]\from [ 0]}{[ 0]\to[
%   %     1]\from[ 1]},
%   % \end{equation*}
%   hence $P\in \mcal C_{2}$. \qed

% \subsection{The non-collapsing hierarchy}

Next we show that the $\mcal C_n$ hierarchy is infinite, by exposing a
sequence of witnesses for $\mcal C_{ 2 n- 1}\subsetneq \mcal C_{ 2 n}$
for all $n\ge 1$.

Let $Q= \idpos 0 0 1$, $P_1= Q\ipomconcat Q$, and for $n\ge 1$,
$P_{ n+ 1}= Q\ipomconcat( P_n\otimes P_n)$.  Note that all these are
series-parallel posets.  Graphically:

\vspace*{-5ex}
\begin{gather*}
  P_1= \pomset{ \cdot \ar[r] & \cdot} \qquad
  P_2= \pomset{ & \cdot \ar[r] & \cdot \\ \cdot \ar[ur] \ar[dr] \\ &
    \cdot \ar[r] & \cdot} \qquad
  P_3= \pomset{ & & \cdot \ar[r] & \cdot \\ & \cdot \ar[ur] \ar[dr] \\ & &
    \cdot \ar[r] & \cdot \\ \cdot \ar[uur] \ar[ddr] \\ & & \cdot
    \ar[r] & \cdot \\ & \cdot \ar[ur] \ar[dr] \\ & &
    \cdot \ar[r] & \cdot} \quad \dotsc
\end{gather*}
\vspace*{-6ex}

\begin{lemma}
  \label{le:bintree}
  $P_n\in \mcal C_{ 2 n}\setminus \mcal C_{ 2 n- 1}$
  % and $P_n\in \mcal D_{ 2 n- 1}\setminus \mcal D_{ 2 n- 2}$
  for all $n\ge 1$.
\end{lemma}

\begin{proof}
  By induction.  For $n= 1$, $P_1\notin \mcal C_1$, but
  $Q\in \mcal C_0\subseteq \mcal C_1$ and hence
  $P_1= Q\ipomconcat Q\in \mcal C_2= \mcal C_1^\ipomconcat$.

  Now for $n\ge 1$, suppose $C_{2n-1} \not \ni P_n\in \mcal C_{ 2 n}$.
  We use Lemma~\ref{le:decomp} to show that
  $P_n\otimes P_n\in \mcal C_{ 2 n+ 1}\setminus \mcal C_{ 2 n}$:
  Obviously
  $P_n\otimes P_n\in \mcal C_{ 2 n+ 1}= \mcal C_{ 2 n}^\otimes$.  If
  $P_n\otimes P_n\in \mcal C_{ 2 n}= \mcal C_{ 2 n- 1}^\ipomconcat$, then
  $P_n\otimes P_n= Q_{1}\ipomconcat\dotsm\ipomconcat Q_k$ for some
  $Q_1,\dotsc, Q_k\in \mcal C_{ 2 n- 1}$.  Yet $P_n\notin \mcal C_1$,
  which contradicts Lemma~\ref{le:decomp}.

  Now to $P_{ n+ 1}= Q\ipomconcat( P_n\otimes P_n)$.  Trivially,
  $P_{ n+ 1}\in \mcal C_{ 2 n+ 2}= \mcal C_{ 2 n+ 1}^\ipomconcat$.  Suppose
  $P_{ n+ 1}\in \mcal C_{ 2 n+ 1}= \mcal C_{ 2 n}^\otimes$.
  $P_{ n+ 1}$ is connected, hence not a parallel product, so that
  $P_{ n+ 1}$ must already be in
  $\mcal C_{ 2 n}= \mcal C_{ 2 n- 1}^\ipomconcat$ and therefore
  $P_{ n+ 1}= R_1\ipomconcat R_2$. Then, by Levi's lemma, there is an iposet
  $S$ such that either $Q= R_1\ipomconcat S$ and
  $S\ipomconcat( P_n\otimes P_n)= R_2$ or $R_1= Q\ipomconcat S$ and
  $S\ipomconcat R_2= P_2\otimes P_n$.  In the second case,
  $S\ipomconcat R_2= P_2\otimes P_n$, which again contradicts
  Lemma~\ref{le:decomp}; in the first case, both $R_1$ and $S$ must be
  single points (with suitable interfaces), so that either
  $R_1= \idpos 0 1 1$ and $R_2= P_{ n+ 1}$ (with an extra starting
  interface) or $R_1= Q$ and $R_2= P_n\otimes P_n$. This shows that
  $P_{ n+ 1}= Q\ipomconcat( P_n\otimes P_n)$ is the only non-trivial
  $\ipomconcat$-de\-com\-po\-si\-tion of $P_{ n+ 1}$. Thus
  $P_n\in \mcal C_{ 2 n- 1}$, a contradiction, and therefore
  $P_{ n+ 1}\notin \mcal C_{ 2 n+ 1}$.  \qed
\end{proof}

\begin{corollary}
  $\mcal C_{ 2 n- 1}\subsetneq \mcal C_{ 2 n}$ for all $n\ge 1$, hence
  the $\mcal C_n$ hierarchy does not collapse, and neither does the
  $\mcal D_n$ hierarchy.
\end{corollary}

\begin{proof}
  The last statement follows from
  $\mcal D_{ 2 n- 2}\subseteq \mcal C_{ 2 n- 1}\subsetneq \mcal C_{ 2
    n}\subseteq \mcal D_{ 2 n+ 1}$. \qed
\end{proof}

\begin{corollary}
  For all $n\in \Nat$, $\bar{ \mcal S}_0\not\subseteq \mcal C_n$ and
  $\bar{ \mcal S}_0\not\subseteq \mcal D_n$.
\end{corollary}

\begin{proof}
  As we have already noted above, $P_n\in \bar{ \mcal S}_0$ for all
  $n$, which together with Lemma~\ref{le:bintree} implies the first
  statement.  The second follows from
  $\mcal C_n\subseteq \mcal D_{ n+ 1}$. \qed
\end{proof}

% \section{The incomplete hierarchy}

We have seen that the $\mcal C_n$ and $\mcal D_n$ hierarchies are
properly infinite and that they contain the set of sp-posets only in
the limit $\bar{ \mcal S}= \bigcup_{ n\ge 0} \mcal C_n= \bigcup_{ n\ge
  0} \mcal D_n$.  

Finally, we turn to the question of characterising this limit
$\bar{ \mcal S}$ geometrically.  Recalling that a poset is
series-parallel iff if it does not contain an induced subposet
isomorphic to \N, we would like a similar characterisation using
forbidden subposets for the gp-(i)posets.  We expose five such
forbidden subposets, but leave the question of whether there are
others to future work.
% We expose below some similar, but weaker, properties of the posets
% in our hierarchy: we show that if a poset is in $\bar{ \mcal S}$,
% then it does not contain an induced \emph{triple-\N}, neither an
% induced \emph{3-crown}, nor an induced \emph{long \N}.  (We will
% define these special posets below.)

Define the following five posets on six points:
% We define three special posets, called respectively the
% \emph{triple-\N}, the \emph{3-crown}, and the \emph{long \N} and given
% as follows:
\begin{gather*}
  \NN= \pomset{ \cdot \ar[r] & \cdot \\ \cdot \ar[r] \ar[ur] & \cdot \\
    \cdot \ar[r] \ar[ur] & \cdot} \qquad%
  \NPLUS = \pomset{ & \cdot \ar[r] & \cdot \\ \cdot \ar[r] \ar[ur] &
    \cdot \\ \cdot \ar[r] \ar[ur] & \cdot} \qquad%
  \NMINUS = \pomset{ & \cdot \ar[r] & \cdot \\ & \cdot \ar[r] \ar[ur]
    & \cdot \\ \cdot \ar[r] & \cdot \ar[ur]} \\
  \TC= \pomset{ \cdot \ar[r]\ar[ddr] & \cdot \\ \cdot \ar[r] \ar[ur] & 
    \cdot \\ \cdot \ar[r] \ar[ur] & \cdot} \qquad%
  \LN= \pomset{ \cdot\ar[r] & \cdot\ar[r] & \cdot \\
  \cdot\ar[r]\ar[urr] & \cdot\ar[r] & \cdot}
\end{gather*}

\begin{proposition}
  \label{pr:forbidden}
  If $P\in \bar{ \mcal S}$, then $P$ does not contain \NN, \NPLUS,
  \NMINUS, \TC, or \LN\ as induced subposets.
\end{proposition}

\begin{proof}
  We only show the proof for \NN; the others are very similar and are
  left to the reader.  We can assume that $P$ is connected.  We use
  structural induction, noting that all $P\in \mcal S$ are \NN-free,
  so it remains to show that $P\ipomconcat Q$ is \NN-free whenever $P$ and
  $Q$ are.
  
  By contraposition, suppose $P\ipomconcat Q$ contains the induced
  sub-\NN\ $\smash{\!\pomset{ a\ar[r] & b \\ c\ar[r]\ar[ur] & d \\
      e\ar[r]\ar[ur] & f}\!}$. Then we show that either $P$ or $Q$
  also have an induced sub-\NN.

  Assume first that $a\in Q$.  Then $a\le_Q b$, hence also $b\in Q$,
  but $b\notin Q_{\min}$, that is, $b\notin s_Q$.  Now
  $e\not\le_{P\ipomconcat Q} b$, which forces $e\in t_P$ and therefore
  in $e\in Q$.  This in turn implies that $d, f\in Q$ and in particular
  $e\le_Q f$.  Thus $f\notin Q_{\min}$ and therefore $f\notin s_Q$,
  which forces $c\in t_P$ and therefore $c\in Q$.  This
  shows that \NN\ lies entirely in $Q$.

 Finally assume that $a\notin Q$. Then $a\in P\setminus t_P$, and as
  $a\not\le_{P\ipomconcat Q} d$ and $a\not\le_{P\ipomconcat Q} f$, we must have $d, f\in
  s_Q$ and therefore $d, f\in P$.  This forces $c, e\in P$ and in particular
  $e\le_P f$.  Thus $e\notin P_{\min}$, whence $e\notin t_P$.
  This in turn forces $b\in s_Q$ and therefore $b\in P$.  This shows
  that  \NN\ lies entirely in $P$. \qed
\end{proof}

\section{Experiments}
\label{se:experi}

We have encoded most of the constructions in this paper with Python to
experiment with gluing-parallel (i)posets.  Notably,
Proposition~\ref{pr:forbidden} is, in part, a result of these
experiments.\footnote{Our software is available at
  \url{http://www.lix.polytechnique.fr/~uli/posets/}} Our prototype is
rather inefficient, which explains why some numbers are ``n.a.'',
\ie~not available, in Table~\ref{ta:numposets}.

Using procedures to generate non-isomorphic posets of different types,
we have used our software to verify that
\begin{enumerate}
\item all posets on five points are in $\bar{ \mcal S}$,
  \ie~gp-posets;
\item \NN, \NPLUS, \NMINUS, \TC, and \LN\ are the only six-point
  posets that are not in $\bar{ \mcal S}$.
\end{enumerate}
We provide tables of gluing-parallel decompositions of posets in
appendix to prove these claims.

We have also used our software to count non-isomorphic posets and
iposets of different types, see Table~\ref{ta:numposets}.  We note
that $\mathsf{P}$ and $\mathsf{SP}$ are sequences no.~A000112
and~A003430, respectively, in the On-Line Encyclopedia of Integer
Sequences (OEIS).\footnote{%
  \label{OEISlabels}%
  See \url{http://oeis.org/A000112}, \url{oeis.org/A003430}, and
  \url{oeis.org/A079566}.} Sequences $\mathsf{GPC}$, $\mathsf{SIP}$,
$\mathsf{IP}$, and $\mathsf{GPI}$ are unknown to the OEIS.

\begin{table}[tbp]
  \centering
  \caption{Different types of posets with $n$ points: all posets;
    % (weakly) connected posets;
    sp-posets; gp-posets; (weakly) connected gp-posets; iposets with
    starting interfaces only; iposets; gp-iposets.}
  \label{ta:numposets}
  % \begin{tabular}{r|rrrrrrrr}
  %   $n$ & $\mathsf{P}( n)$ & $\mathsf{CP}( n)$ & $\mathsf{SP}( n)$ &
  %   $\mathsf{GP}( n)$ & $\mathsf{GPC}( n)$ & $\mathsf{SIP}( n)$ &
  %   $\mathsf{IP}( n)$ & $\mathsf{GPI}( n)$ \\\hline
  %   0 & 1 & 1 & 1 & 1 & 1 & 1 & 1 & 1 \\
  %   1 & 1 & 1 & 1 & 1 & 1 & 2 & 4 & 4 \\
  %   2 & 2 & 1 & 2 & 2 & 1 & 5 & 17 & 16 \\
  %   3 & 5 & 3 & 5 & 5 & 3 & 16 & 86 & 74 \\
  %   4 & 16 & 10 & 15 & 16 & 10 & 66 & 532 & 419 \\
  %   5 & 63 & 44 & 48 & 63 & 44 & 350 & n.a. & 2980 \\
  %   6 & 318 & 238 & 167 & 313 & 233 & n.a. & n.a. & 26566
  % \end{tabular}
  \begin{tabular}{r|rrrrrrrr}
    $n$ & $\mathsf{P}( n)$ & $\mathsf{SP}( n)$ & $\mathsf{GP}( n)$ &
    $\mathsf{GPC}( n)$ & $\mathsf{SIP}( n)$ & $\mathsf{IP}( n)$ &
    $\mathsf{GPI}( n)$ \\\hline
    0 & 1 & 1 & 1 & 1 & 1 & 1 & 1 \\
    1 & 1 & 1 & 1 & 1 & 2 & 4 & 4 \\
    2 & 2 & 2 & 2 & 1 & 5 & 17 & 16 \\
    3 & 5 & 5 & 5 & 3 & 16 & 86 & 74 \\
    4 & 16 & 15 & 16 & 10 & 66 & 532 & 419 \\
    5 & 63 & 48 & 63 & 44 & 350 & n.a. & 2980 \\
    6 & 318 & 167 & 313 & 233 & n.a. & n.a. & 26566
  \end{tabular}
\end{table}

The single iposet on two points which is not gluing-parallel is the
symmetry $[ 2]: 2\to 2$ with $s( 1)= 1$, $s( 2)= 2$, $t( 1)= 2$, and
$t( 2)= 1$.
% $\vcenter{
%   \begin{tikzpicture}[y=.5cm]
%     \node [label=left:{\tiny 1}]  (1) at (0,0) {\inpt};
%     \node [label=right:{\tiny 2}] (2) at (0,0) {\outpt};
%     \node [label=left:{\tiny 2}] (3) at (0,-1) {\inpt};
%     \node [label=right:{\tiny 1}] (4) at (0,-1) {\outpt};
%   \end{tikzpicture}}$.
The prefix of $\mathsf{GP}$ we were able to compute equals the
corresponding prefix of sequence no.~A079566 in the
OEIS,$^{ \ref{OEISlabels}}$
% OEIS~\cite{OEIS-A079566}, 
which counts the number of connected
(undirected) graphs which have no induced 4-cycle $C_4$.  We leave it
to the reader to ponder upon the relation between gp-posets and
$C_4$-free connected graphs.

% \ulilong{%
%   Note that a poset is an interval order iff its induced graph is
%   transitive and its complement $C_4$-free\dots}

\bibliographystyle{myabbrv}
\bibliography{mybib}

\newpage
\appendix

\section*{Appendix}

The following tables show gluing-parallel decompositions of all
(weakly) connected posets on four points, all connected posets on five
points, and all connected posets on six points except for the five
posets \NN, \NPLUS, \NMINUS, \TC, and \LN\ which are not
gluing-parallel.

Given that disconnected posets can be decomposed into posets with
fewer points using $\otimes$ and that all posets on fewer than four
points are series-parallel, hence gluing-parallel, these tables show
the claims in Section~\ref{se:experi}: All posets on five points are
gluing-parallel, as are all but the five exceptional posets \NN,
\NPLUS, \NMINUS, \TC, and \LN\ on six points.

%\footnotesize

%\vspace*{2cm}

\tikzset{->, >=latex', x=1cm, y=.55cm, label
  distance=-.25cm, shorten <=-3pt, shorten >=-3pt}

\begin{longtable}{c|c|c|c}
  \caption{Gluing-Parallel decompositions of connected posets on four
    points} \\
  no. & Poset & \multicolumn{2}{c}{Decomposition} \endhead\hline
  % generated by ramics.py
1 &
\begin{tikzpicture}
\node (1) at (0,0) {\intpt};
\node (2) at (1,0) {\intpt};
\node (3) at (2,0) {\intpt};
\node (4) at (3,0) {\intpt};
\path (1) edge (2);
\path (2) edge (3);
\path (3) edge (4);
\end{tikzpicture}
&
\begin{tikzpicture}
\node (1) at (0,0) {\intpt};
\node (2) at (1,0) {\intpt};
\path (1) edge (2);
\end{tikzpicture}
&
\begin{tikzpicture}
\node (1) at (0,0) {\intpt};
\node (2) at (1,0) {\intpt};
\path (1) edge (2);
\end{tikzpicture}
\\\hline
2 &
\begin{tikzpicture}
\node (1) at (0,0) {\intpt};
\node (2) at (1,0) {\intpt};
\node (3) at (2,0) {\intpt};
\node (4) at (2,-1) {\intpt};
\path (1) edge (2);
\path (2) edge (3);
\path (2) edge (4);
\end{tikzpicture}
&
\begin{tikzpicture}
\node (1) at (0,0) {\intpt};
\node (2) at (1,0) {\intpt};
\path (1) edge (2);
\end{tikzpicture}
&
\begin{tikzpicture}
\node (1) at (0,0) {\intpt};
\node (2) at (0,-1) {\intpt};
\end{tikzpicture}
\\\hline
3 &
\begin{tikzpicture}
\node (1) at (0,0) {\intpt};
\node (2) at (0,-1) {\intpt};
\node (3) at (1,0) {\intpt};
\node (4) at (2,0) {\intpt};
\path (1) edge (3);
\path (2) edge (3);
\path (3) edge (4);
\end{tikzpicture}
&
\begin{tikzpicture}
\node (1) at (0,0) {\intpt};
\node (2) at (0,-1) {\intpt};
\end{tikzpicture}
&
\begin{tikzpicture}
\node (1) at (0,0) {\intpt};
\node (2) at (1,0) {\intpt};
\path (1) edge (2);
\end{tikzpicture}
\\\hline
4 &
\begin{tikzpicture}
\node (1) at (0,0) {\intpt};
\node (2) at (0,-1) {\intpt};
\node (3) at (1,0) {\intpt};
\node (4) at (1,-1) {\intpt};
\path (1) edge (3);
\path (1) edge (4);
\path (2) edge (3);
\path (2) edge (4);
\end{tikzpicture}
&
\begin{tikzpicture}
\node (1) at (0,0) {\intpt};
\node (2) at (0,-1) {\intpt};
\end{tikzpicture}
&
\begin{tikzpicture}
\node (1) at (0,0) {\intpt};
\node (2) at (0,-1) {\intpt};
\end{tikzpicture}
\\\hline
5 &
\begin{tikzpicture}
\node (1) at (0,0) {\intpt};
\node (2) at (1,0) {\intpt};
\node (3) at (1,-1) {\intpt};
\node (4) at (2,0) {\intpt};
\path (1) edge (2);
\path (1) edge (3);
\path (2) edge (4);
\path (3) edge (4);
\end{tikzpicture}
&
\begin{tikzpicture}
\node (1) at (0,0) {\intpt};
\end{tikzpicture}
&
\begin{tikzpicture}
\node (1) at (0,0) {\intpt};
\node (2) at (0,-1) {\intpt};
\node (3) at (1,0) {\intpt};
\path (1) edge (3);
\path (2) edge (3);
\end{tikzpicture}
\\\hline
6 &
\begin{tikzpicture}
\node (1) at (0,0) {\intpt};
\node (2) at (1,0) {\intpt};
\node (3) at (1,-1) {\intpt};
\node (4) at (2,0) {\intpt};
\path (1) edge (2);
\path (1) edge (3);
\path (2) edge (4);
\end{tikzpicture}
&
\begin{tikzpicture}
\node (1) at (0,0) {\intpt};
\end{tikzpicture}
&
\begin{tikzpicture}
\node (1) at (0,0) {\intpt};
\node (2) at (0,-1) {\intpt};
\node (3) at (1,0) {\intpt};
\path (1) edge (3);
\end{tikzpicture}
\\\hline
7 &
\begin{tikzpicture}
\node (1) at (0,0) {\intpt};
\node (2) at (1,0) {\intpt};
\node (3) at (1,-1) {\intpt};
\node (4) at (1,-2) {\intpt};
\path (1) edge (2);
\path (1) edge (3);
\path (1) edge (4);
\end{tikzpicture}
&
\begin{tikzpicture}
\node (1) at (0,0) {\intpt};
\end{tikzpicture}
&
\begin{tikzpicture}
\node (1) at (0,0) {\intpt};
\node (2) at (0,-1) {\intpt};
\node (3) at (0,-2) {\intpt};
\end{tikzpicture}
\\\hline
8 &
\begin{tikzpicture}
\node (1) at (0,0) {\intpt};
\node (2) at (0,-1) {\intpt};
\node (3) at (1,0) {\intpt};
\node (4) at (2,0) {\intpt};
\path (1) edge (3);
\path (2) edge (4);
\path (3) edge (4);
\end{tikzpicture}
&
\begin{tikzpicture}
\node (1) at (0,0) {\intpt};
\node [label=right:{\tiny 1}] (2) at (0,-1) {\outpt};
\end{tikzpicture}
&
\begin{tikzpicture}
\node (1) at (0,0) {\intpt};
\node [label=left:{\tiny 1}] (2) at (0,-1) {\inpt};
\node (3) at (1,0) {\intpt};
\path (1) edge (3);
\path (2) edge (3);
\end{tikzpicture}
\\\hline
9 &
\begin{tikzpicture}
\node (1) at (0,0) {\intpt};
\node (2) at (0,-1) {\intpt};
\node (3) at (1,0) {\intpt};
\node (4) at (1,-1) {\intpt};
\path (1) edge (3);
\path (1) edge (4);
\path (2) edge (4);
\end{tikzpicture}
&
\begin{tikzpicture}
\node (1) at (0,0) {\intpt};
\node [label=right:{\tiny 1}] (2) at (0,-1) {\outpt};
\end{tikzpicture}
&
\begin{tikzpicture}
\node (1) at (0,0) {\intpt};
\node [label=left:{\tiny 1}] (2) at (0,-1) {\inpt};
\node (3) at (1,0) {\intpt};
\path (2) edge (3);
\end{tikzpicture}
\\\hline
10 &
\begin{tikzpicture}
\node (1) at (0,0) {\intpt};
\node (2) at (0,-1) {\intpt};
\node (3) at (0,-2) {\intpt};
\node (4) at (1,0) {\intpt};
\path (1) edge (4);
\path (2) edge (4);
\path (3) edge (4);
\end{tikzpicture}
&
\begin{tikzpicture}
\node (1) at (0,0) {\intpt};
\node (2) at (0,-1) {\intpt};
\node (3) at (0,-2) {\intpt};
\end{tikzpicture}
&
\begin{tikzpicture}
\node (1) at (0,0) {\intpt};
\end{tikzpicture}
\\\hline

\end{longtable}

%\pagebreak

\begin{longtable}{c|c|c|c}
  \caption{Gluing-parallel decompositions of connected posets on five
    points} \\
  no. & Poset & \multicolumn{2}{c}{Decomposition} \endhead\hline
  \input decomp5-edited
\end{longtable}

%\pagebreak

\begin{longtable}{c|c|c|c}
  \caption{Gluing-parallel decompositions of connected gp-posets on
    six points} \\
  no. & Poset & \multicolumn{2}{c}{Decomposition} \endhead\hline
  \input decomp6-edited
\end{longtable}

\end{document}